\PassOptionsToPackage{table}{xcolor}

\documentclass[sigconf,balance=false]{acmart}
\usepackage{popets}

\setcopyright{popets}
\copyrightyear{2024}

\acmYear{2024}
\acmVolume{2024}
\acmNumber{2}

\acmDOI{10.56553/popets-2024-0001}

\acmISBN{}
\acmConference{Proceedings on Privacy Enhancing Technologies}
\usepackage{booktabs} 
\usepackage[ruled]{algorithm2e} 

\SetAlFnt{\small}
\SetAlCapFnt{\small}
\SetAlCapNameFnt{\small}
\SetAlCapHSkip{0pt}
\IncMargin{-\parindent}

\usepackage{array}
\usepackage{graphicx}
\usepackage{comment}
\usepackage{amsmath}
\usepackage[noend]{algorithmic}
\usepackage{hyperref}
\usepackage{xspace}
\usepackage{multirow}
\hypersetup{
     colorlinks = true,
     linkcolor = blue,
     citecolor = blue,
     filecolor = blue,
     urlcolor = blue
     }
\usepackage{cleveref}

\usepackage{thm-restate}

\usepackage{tikz}
\usetikzlibrary{shapes.geometric,arrows.meta,arrows}

\allowdisplaybreaks

\newtheorem{lemma}{Lemma}

\newtheorem{defn}{Definition}
\newtheorem{definition}[defn]{Definition}
\newtheorem{remark}{Remark}

\newtheorem{example}{Example}

\newcommand{\eps}{\varepsilon}

\newcommand{\R}{\mathbb{R}}

\newcommand{\cD}{\mathcal{D}}

\newcommand{\cI}{\mathcal{I}}

\newcommand{\adat}{\cD_{\text{attr}}}

\newcommand{\bD}{\mathbf{D}}

\newcommand{\C}{\mathcal{C}}
\newcommand{\N}{\mathbb{N}}

\newcommand{\Iids}{\mathcal{I}}
\newcommand{\Cids}{\mathcal{C}}
\newcommand{\attr}{\mathtt{a}}
\newcommand{\pr}{\mathsf{pr}}

\newcommand{\ex}[2]{{\ifx&#1& \mathbb{E} \else
\underset{#1}{\mathbb{E}} \fi \left[#2\right]}}
\newcommand{\var}[2]{{\ifx&#1& \mathsf{Var} \else
\underset{#1}{\mathsf{Var}} \fi \left[#2\right]}}

\DeclareMathOperator*{\argmax}{arg\,max}

\newcommand{\nc}{\newcommand}
\nc{\MS}{\mathcal{S}}
\nc{\MP}{\mathcal{P}}

\nc{\cM}{\mathcal{M}}
\newcommand{\valid}{valid\xspace}
\newcommand{\Valid}{Valid\xspace}
\newcommand{\invalid}{invalid\xspace}

\newcommand{\validity}{validity\xspace}
\newcommand{\Validity}{Validity\xspace}

\newcommand{\budget}{contribution bound\xspace}
\newcommand{\Budget}{Contribution Bound\xspace}
\newcommand{\added}[1]{{#1}}

\definecolor{pw}{HTML}{7977B8}
\definecolor{og}{HTML}{3C8031}
\definecolor{maroon}{HTML}{AF3235}
\definecolor{yo}{HTML}{FAA21A}
\definecolor{mybrick}{RGB}{180,14,15}
\definecolor{Gred}{RGB}{219, 50, 54}
\definecolor{Ggreen}{RGB}{60, 186, 84}
\definecolor{Gblue}{RGB}{72, 133, 237}
\definecolor{Gyellow}{RGB}{247, 178, 16}
\definecolor{ToCgreen}{RGB}{0, 128, 0}
\definecolor{myGold}{RGB}{231,141,20}
\definecolor{myBlue}{rgb}{0.19,0.41,.65}
\definecolor{myPurple}{RGB}{175,0,124}

\begin{document}
\startPage{1}

\title{Differentially Private Ad Conversion Measurement}

\author{John Delaney}
\affiliation{%
  \institution{Google}
  \city{} 
  \state{MA} 
  \country{USA} 
}
\email{johnidel@google.com}

\author{Badih Ghazi}
\affiliation{%
  \institution{Google}
  \city{} 
  \state{CA} 
  \country{USA} 
}
\email{badihghazi@gmail.com}

\author{Charlie Harrison}
\affiliation{%
  \institution{Google}
  \city{} 
  \state{TX} 
  \country{USA} 
}
\email{csharrison@google.com}

\author{Christina Ilvento}
\affiliation{%
  \institution{Google}
  \city{} 
  \state{CA} 
  \country{USA} 
}
\email{cilvento@google.com}

\author{Ravi Kumar}
\affiliation{%
  \institution{Google}
  \city{} 
  \state{CA} 
  \country{USA} 
}
\email{ravi.k53@gmail.com}

\author{Pasin Manurangsi}
\affiliation{%
  \institution{Google}
  \city{} 
  \state{} 
  \country{Thailand} 
}
\email{pasin@google.com}

\author{Martin P{\'{a}}l}
\affiliation{%
  \institution{Google}
  \city{} 
  \state{} 
  \country{Switzerland} 
}
\email{mpal@google.com}

\author{Karthik Prabhakar}
\affiliation{%
  \institution{Google}
  \city{} 
  \state{CA} 
  \country{USA} 
}
\email{kprabhakar@google.com}

\author{Mariana Raykova}
\affiliation{%
  \institution{Google}
  \city{} 
  \state{NY} 
  \country{USA} 
}
\email{marianar@google.com}

\renewcommand{\shortauthors}{Delaney et al.}

\begin{abstract}
  In this work, we study \emph{ad conversion measurement}, a central functionality in digital advertising, where an advertiser seeks to estimate advertiser website (or mobile app) conversions attributed to ad impressions that users have interacted with on various publisher websites (or mobile apps).  Using differential privacy (DP), a notion that has gained in popularity due to its strong mathematical guarantees, we develop a formal framework for private ad conversion measurement. In particular, we define the notion of an operationally valid configuration of the attribution rule, DP adjacency relation, contribution bounding scope and enforcement point. We then provide, for the set of configurations that most commonly arises in practice, a complete characterization, which uncovers a delicate interplay between attribution and privacy.
\end{abstract}


\maketitle
\section{Introduction}

Over the last two decades, numerous attacks have illustrated the privacy risks associated with the release of (aggregated, de-identified) information, across various domains \cite{sweeney2015only, narayanan2008robust, homer2008resolving}. This has led to the introduction of differential privacy (DP) \cite{dwork06calibrating,dwork2006our}, a rigorous mathematical notion that quantifies the privacy loss of users against arbitrary adversaries observing the algorithm's output (and regardless of the data contributed by other users). While DP has been deployed in several fields including census data release \cite{abowd2018us}, learning frequently typed words \cite{dp2017learning}, and collection of telemetry data \cite{ding2017collecting}, there has not been any formal study of it for ad conversion measurement---a core functionality in the digital advertising space---that we undertake in this work.

We now provide a quick overview of the problem space.  Let the term \emph{conversion} refer to a valuable action (e.g., a purchase, add-to-cart, newsletter sign up, etc.) on the advertiser website\footnote{Although in this paper we mostly discuss advertiser and publisher \emph{websites}, these could alternatively be \emph{mobile apps}, and the same treatment holds.} and the term \emph{impression} refer to an ad engagement (e.g., a click on an ad or a view of an ad) by the user on the publisher website (and not merely an ad fetch).
In \emph{(ad) conversion measurement}, an advertiser running campaigns on multiple publishers aims to measure the performance of various campaign slices.\footnote{A campaign consists of a subset of impressions, associated with the same campaign ID. Eeach impression is also associated with attributes over which an analyst can slice.} This includes estimating conversion counts and values (the latter can be numerical, e.g., in US dollars) for different settings of impression and conversion features (e.g., Example~\ref{ex:conv_counts_values} below).\  This information can then be used by the advertiser to estimate the conversion rate and the return on ad spend, to optimize the purchase of ad impressions on guaranteed selling channels, and/or to power real-time bidding systems \cite{choi2020online}.
The ease of measuring conversions on advertiser websites (or apps) and attributing them to impressions that the same user had interacted with on publisher websites (or apps) is an important reason why online advertising is significantly more efficient than more traditional forms of advertising (e.g., print media, radio and television).

Privacy is a crucial consideration in conversion measurement. Ad measurement is indeed fundamentally a cross-site functionality, involving multiple publishers and advertisers. For the last two decades, non-private approaches to the problem have allowed third parties to track users across websites; this has been enabled on the Web by third-party cookie technology. However, in recent years, a consensus has emerged that these approaches are too invasive for users, and that new privacy-preserving methods for supporting various ad use cases are critically needed. This has led to the decision to deprecate third-party cookies by several browsers including Apple Safari \cite{safari},  Mozilla Firefox \cite{mozilla}, and Google Chrome \cite{chromium}. Consequently, multiple efforts are currently underway by many browsers, platforms and industry groups to design privacy-preserving APIs that aim to support ad conversion measurement functionalities. \footnote{We note that, on a high-level, the goals of ad conversion measurement are not at odds with privacy; in fact, advertisers typically want to change their ad spend allocation based on patterns that are significant enough to generalize to large groups of users (as opposed to overfitting to individual user actions).} These APIs and proposals include the Interoperable Private Attribution (IPA) proposed by Mozilla and Meta, \cite{mozillaattribution}, Masked LARk from Microsoft \cite{pfeiffer2021masked}, the Privacy Sandbox Attribution Reporting API (ARA) on Chrome \cite{chromeattribution} and Android \cite{androidattribution}, Private Click Measurement (PCM) on Safari \cite{pcm-safari}, SKAdNetwork on iOS \cite{SKAdNetwork} as well as the recent proposal \cite{apple-prio-like} from Apple. While most of these efforts seek to guarantee DP in order to ensure that sensitive user information cannot be recovered from the output of the API, they still lack an end-to-end formal DP  framework. The goal of our work is to develop such a framework, so that proposals in this space are built on a solid mathematical foundation.

\section{Motivation, Setup \& Contributions}

\subsection{Ad Conversion Measurement System}\label{subsec:ad_conv_meas}

We next define the main components of an ad conversion measurement system, and discuss the central concept of \emph{attribution}.

\subsubsection{Basic Definitions.}\label{sec:basic_defs}
An \emph{impression} event consists of (i) a timestamp, (ii) a publisher id, (iii) an advertiser id, (iv) a user id, and (v) metadata associated with the impression (e.g., type indicating if it is a click or view, format indicating size of ad, etc.).\\ 
A \emph{conversion} event consists of (i) a timestamp, (ii) an advertiser id, (iii) a user id, and (iv) metadata associated with the conversion (e.g., a type indicating if it is a purchase or a sign up, a value if it is a purchase etc).\\ 
The entities that are involved in ad conversion measurement are:
\begin{itemize}
\item \emph{Publisher}: a website on which ad impressions take place.
\item \emph{Advertiser}: a website where conversion events take place. 
\item \emph{Ad tech:} entity (often a third party) used by an advertiser to buy, manage, and measure their digital advertising.
\end{itemize}
The most common type of functionality in ad conversion measurement can be thought of as a two-step process: first, conversions are assigned to one or more impressions using a fixed \emph{attribution rule}, and then queries are executed on the resulting \emph{attributed dataset}. On a high level, the attribution rule determines how the credit from a conversion is to be divided over the different ad impressions corresponding to the same advertiser and the same user as the conversion. The attributed dataset consists of (impression, conversion) pairs, each corresponding to an attribution; optionally the pair is also associated with a (fractional) credit.  (In the cases where all credits are equal, we omit them from the attributed dataset notation for simplicity.)  
The queries used in the second step include counting the number of conversions attributed to a subset of impressions (and the corresponding conversion rate), as well as the total return on ad spend for that subset.

\begin{example}\label{ex:conv_counts_values}
An example of the set of impression features can include the publisher website, the advertiser website, the ad engagement type (click or view), as well as its time and geographical location. An example set of conversion features can include the advertiser website, the conversion time, and the conversion type (e.g., add-to-cart, purchase, email sign-up). Thus, examples of ad conversion measurement queries include asking for:
\begin{itemize}
    \item The number of conversions on advertiser1.com attributed to ad impressions on publisher1.com which occurred in the UK.
    \item The total value (in US dollars) of purchases on advertiser1.com taking place on August 31st and attributed to ad impressions on publisher1.com.
\end{itemize},
\end{example}


We next discuss the attribution step in more details.

\begin{figure}[ht]
\centering
\footnotesize
\includegraphics[width=0.5\textwidth]{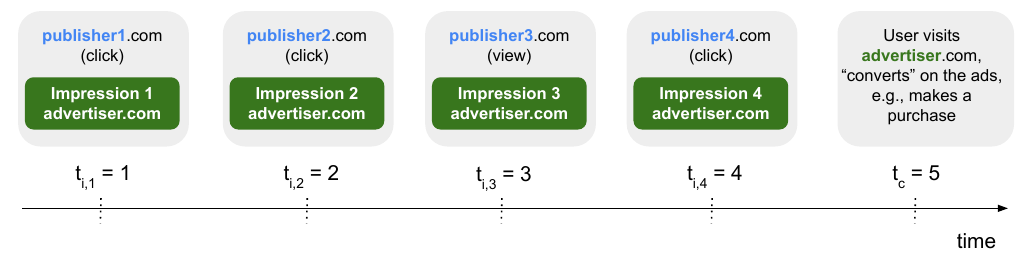}
\caption{Example Attribution Path. In this case, a user interacts with four ad impressions from the same advertiser, but on four different publishers. The third of these interactions is a view, whereas the others are clicks. For simplicity, we assume that the four impressions and the subsequent conversion occur at equally spaced times.}
\label{fig:multi_pub_attribution_path}
\end{figure}

\subsubsection{Attribution Rule}
A key component of any conversion measurement system is the attribution rule. Specifically, consider the setting, illustrated in Figure~\ref{fig:multi_pub_attribution_path}, where a user is exposed to several impressions on multiple publishers before converting on the advertiser website. Which of these impressions should get the credit for driving the conversion? This attribution question has been the subject of numerous studies (see, e.g., \cite{kannan2016path} and the references therein). In practice, the most popular attribution rules include last-touch attribution (LTA), first-touch attribution (FTA), uniform attribution (UNI), and exponential time decayed attribution (EXP). As the names suggest, first- and last-touch attribution assign all the credit to the first and last impressions on the attribution path, respectively, whereas uniform attribution divides the credit evenly among all the impressions, and exponential time decayed attribution assigns to each impression a credit that decays exponentially as a function of the time gap to the conversion. For an illustration 
of these different attribution rules for the example in Figure~\ref{fig:multi_pub_attribution_path}, we refer the reader to Table~\ref{fig:example_credit_division}.

{\small
\begin{table}[H]
\begin{center}
\begin{tabular}{|c|c|c|c|c|} 
  \hline
  {\bf Attribution} & {\bf Credit to}  & {\bf Credit to} & {\bf Credit to} & {\bf Credit to}\\ 
  {\bf Rule}  & {\bf Click on} & {\bf Click on} & {\bf View on} & {\bf Click on} \\ 
    & {\bf publisher1} & {\bf publisher2} & {\bf publisher3} & {\bf publisher4} \\   
  \hline
  LTA & $0$ & $0$ & $0$ & $1$ \\ 
  FTA & $1$ & $0$ & $0$ & $0$ \\ 
  UNI & $0.25$ & $0.25$ & $0.25$ & $0.25$ \\ 
  EXP & $0.0667$ & $0.1333$ & $0.2667$ & $0.5333$ \\ 
  \hline
\end{tabular}
\end{center}
\caption{Example Credit Attribution.}\label{fig:example_credit_division}
\end{table}
}
We next explain how the values in Table~\ref{fig:example_credit_division} were computed. For LTA, all the credit goes to the ``last touch'', which is the user click that took place on publisher4.com. For FTA, all the credit goes to the ``first touch'', which is the user click on publisher1.com. For UNI, the credit is split equally among the $4$ ad impressions ($3$ clicks and $1$ view). For the EXP attribution rule, we assume that the half-life parameter\footnote{We refer the reader to Section~\ref{subsubsec:attr_logic} for the formal definition.} is $1$, and that the times of the successive impressions and conversion are all separated by $1$ time unit, as shown in Figure~\ref{fig:multi_pub_attribution_path}; thus, the credit assigned to ad impression $i$ is $0.5^{5-i}/(0.5 + 0.5^2 + 0.5^3 + 0.5^4)$ for each $i \in \{1, 2, 3, 4\}$.

In practice, different ad techs might offer advertisers support different attribution rules to be used when measuring conversions. For instance, while FTA might be desirable for some advertised products, LTA might be natural in other settings. In this work, we consider all of the basic and most popular attribution rules that were mentioned above.\footnote{In addition to the aforementioned attribution rules, we also cover position-based, U-shaped, and impression-priority attribution; see Section~\ref{subsubsec:attr_logic} for the formal definitions.}

\subsection{Privacy Model}\label{subsec:privacy_model}

We start by describing the threat model for privacy-preserving ad conversion measurement.
\subsubsection{Threat Model}

To illustrate the threat model, we consider the setting where an untrusted third party (e.g., an ad tech as defined in Section~\ref{sec:basic_defs}) would like to run  queries on a conversion measurement dataset pertaining to multiple publishers and advertisers.
We assume that a central trusted curator is in charge of executing the query on the dataset and sending the output to the ad tech; the trusted curator could, e.g., be a web browser or mobile platform, a trusted third party, or a secure multi-party computation protocol. The goal is to protect sensitive information (e.g., related to a single impression, conversion, or user?) from being leaked to the ad tech through the output of the protocol. Examples of sensitive information in this setting include app or web browsing history (e.g., did the user recently visit a sensitive website), or the shopping history of the user. To protect the data of individual users from leaking through an output released to an untrusted third party, DP has become the gold standard. It has been suggested as a primary privacy guardrail in multiple industry proposals for privacy-preserving ad measurement systems. Therefore, in this work, we study DP as the desired privacy guarantee on the output of the ad measurement system.
\begin{remark} (Robustness to Side Information)
We note that in some cases, depending on browser constraints, user sign-ins, and business arrangements, the third party could have a prior partial view of the dataset.  E.g., it could know the set of all impressions (across all users but without knowing the associated user id for each impression), or the set of all conversions, or both. Nevertheless, the protection offered by DP would still be meaningful in this setting, since DP is robust to the presence of side information.
\end{remark}

\subsubsection{Differential Privacy Ingredients}

\paragraph{Adjacency Relation}
On a high level, DP dictates that the distributions of the output of a (randomized) algorithm on two \emph{adjacent} conversion measurement datasets are statistically indistinguishable (see Section~\ref{sec:dp_def} for a formal definition).
It is thus necessary to specify the adjacency relation (a.k.a. privacy unit) to which the DP definition applies. Due to the highly fragmented nature of conversion measurement datasets (with the conversion taking place on one of the advertisers and the impressions taking place on different publishers), it turns out there are multiple natural alternatives for defining the adjacency relation, with subtle implications on the privacy-utility trade-offs. The different possibilities are listed in Table~\ref{fig:adj_table}. For each adjacency relation, the allowed difference between two adjacent datasets consists of all user engagements that only belong to a single tuple in the adjacency relation. For instance, for the \textsl{user $\times$ advertiser} relation, two adjacent datasets can differ on the set of all impressions of Alice associated with advertiser1.com (and shown on any publisher), along with all conversions of Alice on advertiser1.com; this is because all of these engagements are only associated with a single (user, advertiser) tuple, namely, (Alice, advertiser1.com). On the other hand, for the \textsl{user $\times$ publisher $\times$ advertiser} relation, the difference consists of all impressions associated with a fixed publisher and a fixed advertiser (e.g., all impressions of Alice shown on publisher1.com, and that are associated with advertiser1.com); note that the conversions associated with the advertiser (e.g., advertiser1.com) are not included here, as they can be associated with multiple other publishers (e.g., publisher2.com), and are thus also related to other (user, publisher, advertiser) tuples in the adjacency relation, e.g., (Alice, publisher2.com, advertiser1.com) is such a tuple.
\begin{remark}[Intuitive Interpretation of the Different Adjacency Relations]
The different adjacency relations described in Table~\ref{fig:adj_table} offer a spectrum of possible privacy guarantees. On one end, the \textsl{conversion} (respectively, \textsl{impression}) adjacency relation seeks to protect a single conversion (respectively, impression) from being leaked by the output of the algorithm. On the other end, the \textsl{user} relation protects any information related to all the impressions and conversions of any user from being leaked. The other notions can be seen as interpolating between these two ends. E.g., the \textsl{user $\times$ advertiser} relation protects all the user’s impressions and conversions pertaining to a single advertiser, but it does not necessarily protect information that can be deduced by observing the user’s impressions and conversions across multiple advertisers. In particular, under this \textsl{user $\times$ advertiser} relation, observing the output of the privacy-preserving ad conversion measurement the system would not substantially increase an attacker's ability to distinguish whether Alice had any attributed conversions associated with advertiser1.com. The choice of the adjacency relation depends on the privacy protection that the system designer seeks to guarantee, and the associated utility trade-offs that they are willing to accept. E.g., the \textsl{user $\times$ advertiser} adjacency can be natural in some settings as it would prevent the ad-tech associated with a publisher from learning that a visitor to the publisher later visited a (sensitive) advertiser site. On the other hand, a \textsl{user $\times$ publisher} adjacency relation can prevent an ad-tech associated with the advertiser from learning about the actions of a converting user on publisher sites.
\end{remark}

\begin{remark}[Time Dimension in Adjacency Relations]
    In practice, it is common to include the time dimension in some adjacency relations, in particular, the \textsl{user}, \textsl{user $\times$ publisher} and \textsl{user $\times$ advertiser} ones. Moreover, DP composition can be used to handle the case where the time dimension in the adjacency relation can cover multiple releases of the output of the system on different subsets of the dataset.
\end{remark}


\begin{figure*}[t]
\centering
\footnotesize
\includegraphics[width=1\textwidth]{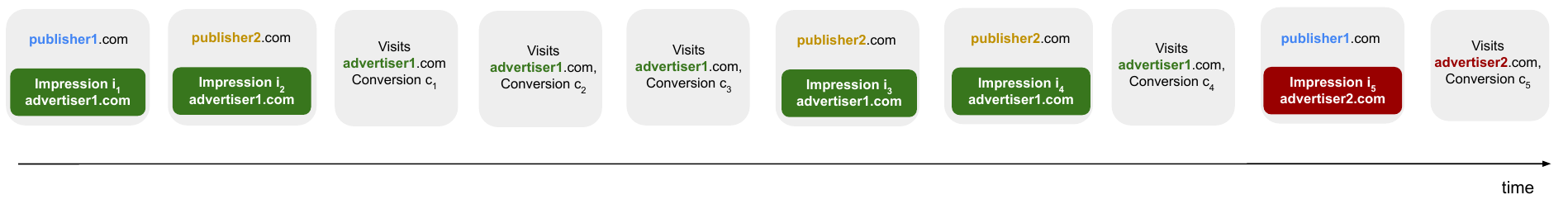}
\caption{Attribution Path for Multiple Publishers and Advertisers. In this example, the user interacts with ads on two publishers, and converts on the two corresponding advertiser websites.}
\label{fig:mult_publishers_advertisers_contribution_capping}
\end{figure*}

\paragraph{Contribution Bounding Scope}
To ensure privacy for a given adjacency relation, a core component of DP systems is the notion of a \budget. Note that, in the case of ad conversion measurement, the number of interactions can naturally be unbounded. For example:
\begin{itemize}
    \item An impression could lead to multiple conversions.
    \item Different impressions can be on the attribution path of the same conversion.
    \item The same user can be shown many impressions on the same or on different publishers, and could convert multiple times on the advertiser.
\end{itemize}
See Figure~\ref{fig:mult_publishers_advertisers_contribution_capping} for an example. To be able to guarantee a bounded privacy leakage, the contributions of the interactions to the computed function should be restricted by a certain \emph{\budget}.
The \emph{contribution bounding scope} is the set of user interactions that share the same contribution bound. The scope can be any one of the options listed in Table~\ref{fig:adj_table} for the adjacency relation. We will consider in this work the natural setting where the contribution bounding scope is the same as the adjacency relation.
It turns out that each adjacency relation / contribution bounding scope can have a different interplay with the attribution rule, and with the resulting DP guarantee. Moreover, as we will see shortly, some will be easier to operationalize than others.

\begin{table*}[h!]
\begin{center}
\begin{tabular}{ | m{4.15cm} |m{9cm}| } 
  \hline
  {\bf Adjacency Relation} & {\bf Difference between Adjacent Datasets} \\ 
  \hline
  Impression & A single impression \\ 
  \hline
  Conversion & A single conversion \\ 
  \hline
  User $\times$ Publisher & All impressions shown to a user on a given publisher \\ 
  \hline
  User $\times$ Advertiser & All impressions shown to a user and for a given advertiser, and all conversions by the same user on the same advertiser \\   
  \hline
  User $\times$ Publisher $\times$ Advertiser & All impressions shown to a user on a given publisher and corresponding to a given advertiser  \\     
  \hline
  User & All impressions shown on all publishers, and all conversions occurring on all advertisers, for a given user \\ 
  \hline
\end{tabular}
\end{center}
\caption{DP Adjacency Relations.}\label{fig:adj_table}
\end{table*}

\paragraph{\Budget Enforcement}
One particular choice, that turns out to be important, is whether the \budget should be enforced before (\emph{pre-attribution contribution capping}) or after (\emph{post-attribution contribution capping})  the attribution rule is applied to the dataset. In the former  case, only impressions belonging to a scope with a non-zero remaining \budget can be considered on the attribution path of a conversion. In the latter  case, a conversion could get attributed to an impression with a remaining \budget of $0$, only to get discarded at the \budget enforcement step (without the attribution falling back to any other impression on the path). In other words, pre-attribution \budget enforcement limits the number of impressions or conversions that belong to the contribution bounding scope and that can \emph{enter} the attribution rule. By contrast, post-attribution enforcement limits the number of post-attribution (impression, conversion) pairs that belong to the contribution bounding scope.\footnote{In the case of multi-touch attribution, each (impression, conversion) in the attributed dataset is weighted, and when enforced post-attribution, the \budget in fact limits the \emph{total weight} of the pairs that belong to the contribution bounding scope.} See Table~\ref{tab:post_attribution_capping_example} for the attributed dataset resulting from applying post-attribution contribution bounding with a bound of $2$ to the dataset given in Figure~\ref{fig:mult_publishers_advertisers_contribution_capping} and with the LTA rule. 

\begin{table}[H]
\begin{tabular}{|c|c|}
\hline
{\bf Contribution Bounding} & {\bf Attributed Dataset} \\
\hline
& ($i_2$, $c_1$) \\
& ($i_2$, $c_2$) \\
None & ($i_2$, $c_3$) \\
& ($i_4$, $c_4$) \\
& ($i_5$, $c_5$) \\
\hline
& ($i_2$, $c_1$) \\
Impression & ($i_2$, $c_2$) \\
(\Budget = 2) & ($i_4$, $c_4$) \\
& ($i_5$, $c_5$) \\
\hline
User $\times$ Advertiser
& ($i_2$, $c_1$) \\
(\Budget = 2) & ($i_2$, $c_2$) \\
& ($i_5$, $c_5$) \\
\hline
User & ($i_2$, $c_1$) \\
(\Budget = 2) & ($i_2$, $c_2$) \\
\hline
\end{tabular}
\caption{Post-Attribution Contribution Bounding. The input dataset is shown in Figure~\ref{fig:mult_publishers_advertisers_contribution_capping}. The LTA rule was applied.}
\label{tab:post_attribution_capping_example}
\end{table}

We now explain how the attributed datasets were generated in Table~\ref{tab:post_attribution_capping_example}. First, note that in the absence of any contribution bounding, each conversion should simply be attributed to the last impression occurring prior to it. So in the dataset given in Figure~\ref{fig:mult_publishers_advertisers_contribution_capping}, conversions $c_1$, $c_2$ and $c_3$ should be attributed to impression $i_2$, conversion $c_4$ should be attributed to impression $i_4$, and conversion $c_5$ should be attributed to impression $i_5$. This explains the first row of Table~\ref{tab:post_attribution_capping_example}. We note that post-attribution contribution bounding with a per conversion contribution bounding scope is a no-op, so it would result in the same attributed dataset as the first (i.e., ``None'') row in Table~\ref{tab:post_attribution_capping_example}. In second row of Table~\ref{tab:post_attribution_capping_example}, a post-attribution contribution bound of $2$ is applied for each impression. This leads to dropping the pair $(i_2, c_3)$ from the attributed dataset because conversions $c_1$ and $c_2$ are already attributed to impression $i_2$. In the third row of Table~\ref{tab:post_attribution_capping_example}, a contribution bound of $2$ is applied post-attribution to each (user, advertiser) pair. Since Figure~\ref{fig:mult_publishers_advertisers_contribution_capping} specifies a dataset for a single user and for two advertisers (namely, advertiser1.com and advertiser2.com), this entails capping the number of attributed conversions for of each of the two advertisers to $2$. Compared to the second row, this has the additional effect of dropping the pair $(i_4, c_4)$ from the attributed dataset, since conversions $c_1$ and $c_2$ occur on the same advertiser and already appear in the attributed dataset. Finally, in the last row of Table~\ref{tab:post_attribution_capping_example}, a contribution bound of $2$ is applied post-attribution for each user. For the user whose dataset is shown in Figure~\ref{fig:mult_publishers_advertisers_contribution_capping}, this implies that the total number of conversions (across all advertisers) appearing in the attributed dataset should be bounded to at most $2$. Compared to the third row of Table~\ref{tab:post_attribution_capping_example}, this has the further effect of dropping the pair $(i_5, c_5)$ from the attributed dataset, because the same user already has two conversions, $c_1$ and $c_2$, appearing in the attributed dataset (despite the fact that these conversions occur on a different advertiser).

For the results for pre-attribution contribution bounding, see Table~\ref{tab:pre_attribution_capping_example}. 

\begin{table}[ht]
\begin{tabular}{|c|c|}
\hline
{\bf Contribution Bounding} & {\bf Attributed Dataset} \\
\hline
User $\times$ Advertiser & ($i_5$, $c_5$) \\  (\Budget = 2) & \\ 
\hline
User & Empty \\ 
(\Budget = 2) & \\
\hline
\end{tabular}
\caption{Pre-Attribution Contribution Bounding. The input dataset is shown in Figure~\ref{fig:mult_publishers_advertisers_contribution_capping}. The LTA rule was applied.}
\label{tab:pre_attribution_capping_example}
\end{table}

We next explain how the attributed datasets were generated in Table~\ref{tab:pre_attribution_capping_example}. In the first row, a contribution bound of $2$ is applied \emph{pre-attribution} for each (user, advertiser) pair. For the single user corresponding to Figure~\ref{fig:mult_publishers_advertisers_contribution_capping}, the first advertiser will have its contribution bound exhausted after impressions $i_1$ and $i_2$ are processed; hence, no conversion occurring on the first advertiser will appear in the attributed dataset. By contrast, for the second advertiser, impression $i_5$ and and conversion $c_5$ will be processed, and the resulting pair $(i_5, c_5)$ will be added to the attributed dataset (and at that point the contribution bound for the second advertiser would be exhausted). For the last row of Table~\ref{tab:pre_attribution_capping_example}, the contribution bound of $2$ is enforced pre-attribution at the user level. For the user corresponding to Figure~\ref{fig:mult_publishers_advertisers_contribution_capping}, the contribution bound would be exhausted after impressions $i_1$ and $i_2$ are processed, and thus the attributed dataset would be empty. We point out that pre-attribution contribution bounding with an per impression contribution bounding scope is a no-op; hence it results in the same attributed dataset as the first (i.e., ``None'') row in Table~\ref{tab:post_attribution_capping_example}.

As is the case in Tables~\ref{tab:post_attribution_capping_example} and~\ref{tab:pre_attribution_capping_example}, pre-attribution contribution bounding in general results in a larger signal loss in the attributed dataset compared to post-attribution contribution bounding. As we will discuss shortly, the design choice of whether to enforce the \budget pre- or post-attribution significantly affects the end-to-end privacy of the ad conversion measurement system.

\subsection{Valid Configurations}
The DP aspects of a private conversion measurement system are mostly captured by the choice of (i) the attribution rule, (ii) the adjacency relation, (iii) the contribution bounding scope, and (iv) the \budget enforcement point. We refer to a setting of each of these choices as a \emph{configuration}. It is natural to consider a configuration to be operationally \emph{valid} if for every positive integer $r$, enforcing a \budget of $r$ at the required point ensures that any two adjacent datasets always result in two post-attribution post-enforcement datasets that differ on at most $C_0 \cdot r$ many (impression, conversion) pairs, where $C_0$ is an absolute constant independent of the numbers of publishers and advertisers.\footnote{To cover multi-touch attribution, we in fact bound the $\ell_1$-distance between the two post-attribution post-enforcement datasets of weighted (impression, conversion) pairs. See Definitions~\ref{def:attr-dataset} and~\ref{def:valid_configs} for more details.} If this property does not hold, then a change of, e.g., a single impression's contributions, within the the \budget of $r$, could result in a change in the attributed dataset of magnitude growing with the (unbounded and potentially very large) number of publishers and advertisers showing ads to a given user. It turns out that this condition is not only sufficient to ensure the DP of the ad measurement system, but also in a sense necessary, unless the noise is increased with the total number of publishers or advertisers---which is practically unwieldy as this number is not fixed and can vary from user to user (note that the subset of publishers and advertisers showing ads to a given user as they browse the Web cannot be fixed ahead of time). In other words, a configuration is deemed \emph{invalid} if the sensitivity increases as the number of advertisers or publishers increases. For more details, we refer the reader to Lemma~\ref{lem:main-dp} and the paragraph following it.
\subsection{Our Contributions}

In addition to formally defining the framework for ad conversion measurement and defining the notion of an operationally valid configuration, we provide a complete characterization of the validity of the configurations that most commonly arise in practice.

\paragraph{Classification}
We provide a complete classification of all the configurations of attribution rule, adjacency relation and \budget enforcement point, that are operationally valid; see Table~\ref{tab:conv-level}. We discuss the obtained classification next.

We first note that pre-attribution \budget enforcement results in valid configurations for all considered attribution rules and adjacency relations. A possible challenge to such an enforcement point is that an impression or publisher can incur a deduction from their \budget whenever they are part of the \emph{input} to the attribution rule, and even if they are not selected for attribution. This can result in situations where a publisher can see their \budget totally exhausted due to conversions that got attributed to other publishers. This is in fact the main motivation for considering post-attribution \budget enforcement, which we discuss next.

It turns out that the adjacency relations that are valid for all attribution rules in the case of post-attribution enforcement are the \textsl{conversion}, \textsl{user $\times$ advertiser}, and \textsl{user} options. A limitation of the conversion adjacency relation is that the privacy would degrade as a user converts more than once. Given that conversion events are in practice not restricted to purchases (e.g., page views, email signups, and add-to-carts can also qualify as conversions), the privacy leakage could increase noticeably. On the other hand, the \textsl{user} adjacency relation could be operationally challenging to enforce as all the publishers and advertisers would have to share the same \budget, which could end up being dominated by certain publishers and/or advertisers. For the other relations of \textsl{impression}, \textsl{user $\times$ publisher}, and \textsl{user $\times$ publisher $\times$ advertiser}, we demonstrate in Table~\ref{tab:conv-level} that the situation is much more delicate, as the validity of the different attribution rules turns out to depend on the adjacency relation. For instance, we show that, surprisingly, while the \textsl{impression}, and \textsl{user $\times$ publisher $\times$ advertiser} adjacency relations and contribution bounding scopes admit valid configurations for post-attribution \budget enforcement, the \textsl{user $\times$ publisher} adjacency relation does not.

Our results suggest that if all the considered attribution rules are to be supported, then either pre-attribution enforcement, or a \textsl{user}, \textsl{user $\times$ advertiser}, or \textsl{conversion} adjacency relation should be used. If, however, post-attribution enforcement is desired and a middle ground is sought between the \textsl{conversion} and \textsl{user} contribution bounding scopes, then only a subset of the attribution rules can be supported (as in Table~\ref{tab:conv-level}).

\added{
We emphasize that the dimensions considered in our classification are fundamental to \emph{any} conversion measurement system. Specifically, any such system has to select an attribution rule. Moreover, any DP implementation has to choose a privacy unit. It also has to bound contributions, and keep track of a remaining contribution bound.
}

For a high-level overview of the proofs, we refer the reader to Section~\ref{sec:pf_overview}.

\added{
We next give an example illustrating the idea captured by the notion of \emph{invalid} configurations.
\paragraph{Example of an Invalid Configuration}
Consider the configuration where LTA is selected as the attribution rule,  the adjacency relation (and the contribution bounding scope) is set to \textsl{user} $\times$ \textsl{publisher}, and contribution bounding is performed post-attribution.
Moreover, consider the typical setting where an advertiser runs a campaign displaying ads on multiple publishers. The advertiser’s goal is to estimate the number of conversions attributed to ad impressions shown on each publisher. Since the adjacency relation is \textsl{user} $\times$ \textsl{publisher}, and since summation has sensitivity $O(1)$, one would hope for an $\epsilon$-DP algorithm with error $O(1/\epsilon)$.
On a high level, our results suggest that surprisingly, this is not possible to achieve since adding last-touch interactions on a publisher can remove attributed conversions for \emph{all} other publishers. Hence, the sensitivity in fact grows with the (practically unbounded) number of publishers, which would result in poor measurements even if a publisher has thousands of attributed conversions.
}

\subsection{Additional Related Work}
There have been several previous works on (non-private) conversion measurement, e.g., \cite{agarwal2010estimating, menon2011response, lee2012estimating, rosales2012post}. We point out that it is common in the literature on DP to consider notions between protecting a single contribution and protecting all the user’s contributions; see, e.g., \cite{pejo2022neighborhood} and the references therein. We also note that some recent ad conversion measurement systems rely on ad-hoc privacy notions; see, e.g., \cite{ayala2022show}. To the best of our knowledge, our work is the first study of the end-to-end (differential) privacy of ad conversion measurement systems.

The very recent works of \cite{dawson2023optimizing} and \cite{aksu2023summary} provide an empirical evaluation of a differentially private ad conversion measurement system similar to the one studied in this work. Their focus is the Privacy Sandbox Attribution Reporting API (ARA) on Chrome and Android. They consider last touch attribution and an \textsl{impression} privacy unit. Both of these work consider the linear queries functionality (e.g., conversion counts and values). The former focuses on hierarchical queries whereas the latter studies the non-hierarchical setting. These works provide a concrete instantiation of a DP ad conversion measurement system similar to the one studied in this work, and they empirically evaluate the error, on real ad conversion datasets, for different values of the differential privacy parameter $\epsilon$. We refer the reader to these two papers for more details. Note that, in our terminology, ARA assumes that the attributed dataset is the input, and that post-attribution capping and noising is performed. Thus, our work complements these previous works: the valid configurations (for the \textsl{impression} adjacency relation) in our work imply that ARA satisfies an end-to-end DP guarantee for the corresponding attribution rules. Meanwhile, the invalid configurations imply that the end-to-end DP guarantee may not hold for those attribution rules.

\paragraph{Organization of Rest of the Paper}

We start Section~\ref{sec:prelims} with some notation that will be used in the rest of the paper. We recall the formal definition of DP in Section~\ref{sec:dp_def}.
In Section~\ref{subsubsec:attr_logic}, we formally define the various attribution rules that will be studied in this paper. In Section~\ref{sec:dp_meas_system}, we present the notion of an operationally valid configuration, along with its connection to the design of a DP ad conversion measurement system. The pre- and post-attribution \budget enforcement algorithms are described in Section~\ref{sec:budget-enforcement}. We present our main validity and invalidity results in Section~\ref{sec:our-results}. Some of the proofs are given in Section~\ref{sec:proof} (with the rest deferred to the Appendix). Our work opens up several interesting areas of exploration; we describe some of these in Section~\ref{conc_future_directions} \added{where we also discuss our results in the context of related practical applications}.

\newcommand{\wline}{\arrayrulecolor{white}}
\newcommand{\bline}{\arrayrulecolor{black}}
\newcommand{\gline}{\arrayrulecolor{green}}
\newcommand\crule[3][black]{\textcolor{#1}{\rule{#2}{#3}}}

\begin{table*}[h!]
\small
\setlength\arrayrulewidth{1.1pt}
\wline
\begin{center}
\begin{tabular}{r | c | c | c | c | c | c | c |}
\bline 
\cline{2-8}
& LTA & FTA & UNI & EXP & U-S & POS & IPA \\
\cline{2-8}
& \multicolumn{7}{c}{\textsf{Post-Attribution}} \\
\cline{2-8}
\wline
 Impression & \cellcolor{green!30}
 Thm.~\ref{thm:post-attr-impression-lta}($+$) & \cellcolor{green!30}
 Thm.~\ref{thm:post-attr-impression-fta}($+$) & \cellcolor{red!30}
 Thm.~\ref{thm:post-attr-impression-uni}($-$) & \cellcolor{red!30} Cor.~\ref{cor:post-attr-impression-exp}($-$) & \cellcolor{red!30}
 Thm.~\ref{thm:post-attr-impression-ushaped}($-$) & \cellcolor{orange!30} ($\pm$) & \cellcolor{orange!30} ($\pm$) \\ 
\hline
User & \multicolumn{7}{c}{\cellcolor{green!30} Thm.~\ref{thm:post-attr-user}($+$)}  \\
\hline
User $\times$ Publisher & \multicolumn{7}{c}{\cellcolor{red!30} Thm.~\ref{thm:post-attr-user-publisher}($-$)} \\
\hline
User $\times$ Advertiser & \multicolumn{7}{c}{\cellcolor{green!30} Thm.~\ref{thm:post-attr-user-advertiser}($+$)} \\
\hline
User $\times$ Pub $\times$ Adv & \cellcolor{red!30} 
Thm.~\ref{thm:post-attr-publisher-ad-lta}($-$) & \cellcolor{green!30} 
Thm.~\ref{thm:post-attr-publisher-ad-fta}($+$) &
\multicolumn{3}{c}{\cellcolor{red!30} Cor.~\ref{cor:post-attr-publisher-ad-uni-exp-ushaped}($-$)} 
& \cellcolor{orange!30} ($\pm$)
& \cellcolor{orange!30} ($\pm$)
\\
\hline
& \multicolumn{7}{c|}{\textsf{Pre-Attribution}} \\
\hline 
 Impression
 & \multicolumn{7}{c}{\cellcolor{green!30}} \\
\cline{2-8}
User & 
\multicolumn{7}{c}{\cellcolor{green!30}} \\
\cline{2-8}
User $\times$ Publisher &  \multicolumn{7}{c}{\cellcolor{green!30} Thm.~\ref{thm:pre-attr}($+$)} \\
\cline{2-8}
User $\times$ Advertiser & \multicolumn{7}{c}{\cellcolor{green!30} } \\
\cline{2-8}
User $\times$ Pub $\times$ Adv & \multicolumn{7}{c}{\cellcolor{green!30}} \\
\wline \hline \hline
Conversion* & \multicolumn{7}{c}{\cellcolor{green!30} Thm.~\ref{thm:conv-valid}($+$)} \\
\hline
\bline
\cline{2-8}
\end{tabular}
\end{center}
\caption{
\Validity of post- and pre-attribution enforcement configurations.  A  \crule[green!30]{1.5em}{.6em}($+$) cell indicates a \valid configuration, a \crule[red!30]{1.5em}{.6em}($-$) cell indicates an \invalid configuration, and a
\crule[orange!30]{1.5em}{.6em}($\pm$) cell means means that there are attribution rules in that family that result in a \valid configuration and an \invalid configuration. Specifically, both the IPA class and the POS class contain FTA and UNI; in the \textsl{impression} and \textsl{user $\times$ publisher $\times$ advertiser} adjacency relations, the former is \valid but the latter is \invalid.  (*) For the \textsl{conversion} adjacency relation, no \budget enforcement is applied as the conversion is already only used once in the attribution rule.} \label{tab:conv-level}
\end{table*}

\section{Preliminaries}\label{sec:prelims}

\paragraph{Notation}
For any positive integer $n$, we denote by $[n]$ the set $\{1, \dots, n\}$. For any finite set $S$, we denote by $S^*$ the set of all finite-length non-empty sequences of elements of $S$. For any positive integer $m$, the $m$-dimensional probability simplex, denoted by $\Delta_m$, is defined as the set of all vectors in $[0,1]^{m+1}$ whose coordinates add up to $1$. The $\ell_1$-norm of a vector $v \in \mathbb{R}^d$ is defined as $\|v\|_1 = \sum_{i=1}^d |v_i|$. The Laplace distribution with zero mean and scale parameter $b > 0$ is the continuous probability distribution whose probability density function is given by $f(x; b) = \frac{1}{2b} \cdot e^{-\frac{|x|}{b}}$, for any real number $x$.

\subsection{Differential Privacy}\label{sec:dp_def}
We denote two adjacent datasets $\bD$ and $\bD'$ by $\bD \sim \bD'$. The adjacency notions considered in this work are listed in Table~\ref{fig:adj_table}, and will be further discussed in Section~\ref{sec:adjacency_and_privacy_scope}, but DP can be defined generally for any such relation. 

\begin{definition}[Differential Privacy \cite{dwork06calibrating}]\label{def:dp}
Let $\epsilon \geq 0$.
A randomized mechanism $\cM$ is 
\emph{$\epsilon$-differentially private (denoted by $\epsilon$-DP)}
if for each pair $\bD \sim \bD'$ of adjacent datasets and each subset $\MS$ of outputs of $\cM$, it holds that $\Pr[\cM(\bD) \in \MS] \leq e^\epsilon \cdot \Pr[\cM(\bD') \in \MS]$,
where the probabilities are over the randomness in $\cM$. 
\end{definition}
Intuitively, the DP definition guarantees that the outputs of two adjacent datasets are approximately statistically indistinguishable. The degree of indistinguishibility is dictated by the $\epsilon$ parameter. The smaller $\epsilon$ is, the more private the algorithm would be. In our setting, the dataset $\bD$ consists of the impressions and conversions (pre-attribution), across all users, advertisers and publishers. The output $\cM(\bD)$ is the output of the privacy-preserving ad conversion measurement system.

DP satisfies several useful
mathematical properties that have made it an appealing measure of privacy. These include robustness to post-processing, composition, and group privacy. For a comprehensive overview of the area, we refer the reader to the monographs \cite{dwork2014algorithmic, vadhan2017complexity}.

\subsection{Attribution Rule}\label{subsubsec:attr_logic}
The \emph{attribution rule} function (see, e.g., \cite{IAB} for background) takes as input a sequence of $m$ impressions and a conversion, all corresponding to the same user and advertiser, and returns a fraction in $[0,1]$ for each of the $m$ impressions. We denote this function by $\attr: \Iids^* \times \Cids \to [0,1]^*$, where $\Iids$ is the set of all possible impressions, and $\Cids$ is the set of all possible conversions.
We assume that for any input $((i_1, \dots, i_m), c)$ to the attribution function $\attr$, it is the case that the impressions $i_1, \dots, i_m$ have been sorted from least to most recent according to their timestamps and $c$ occurs later than $i_m$.  Moreover, it is assumed that $\attr((i_1, \dots, i_m), c) \in \Delta_{m-1}$.

\subsubsection{Single-Touch}
In \emph{single-touch} attribution, only a single coordinate in the output $\attr((i_1, \dots, i_m), c)$ is equal to $1$ and all the other $m-1$ coordinates are equal to $0$. We next describe some notable special cases of single-touch attribution.

\paragraph{Last-Touch Attribution (LTA)} $\attr((i_1, \dots, i_m), c) = (0, \dots, 0, 1)$, i.e., the last impression in the sequence is selected.

\paragraph{First-Touch Attribution (FTA)} 
$\attr((i_1, \dots, i_m), c) = (1, 0, \dots, 0)$, i.e., the first impression in the sequence is selected.

\subsubsection{Multi-Touch}
While the single touch attribution rules assign all the credit to a single impression, \emph{multi-touch} attribution allows spreading the credit over more than one impression. The simplest multi-touch attribution rule is uniform (aka linear).

\paragraph{Uniform (UNI)} $\attr((i_1, \dots, i_m), c) = (1/m, \dots, 1/m)$, i.e., the credit is split equally among all the impressions.

\paragraph{Exponential Time Decay (EXP)}
In the EXP rule with half-life parameter $t_{1/2}$, the credit assigned to a given impression is
proportional to $(0.5)^{\frac{t}{t_{1/2}}}$, where $t$ is the difference between the timestamp of the impression and that of the conversion.
 
\paragraph{U-Shaped (U-S)}
If there are at least three impressions, $40\%$ of the credit goes to the first touch, $40\%$ of the credit goes to the last touch, and the remaining credit is divided uniformly over all the intermediate impressions (i.e., those that are neither first nor last). If there are two impressions, we assume that the credit is split equally between them.

\paragraph{Positional (POS)}
In positional (aka position based) attribution, the credit assigned to each impression is based on the total number of impressions and the order in which this impression occurs, i.e., the credit does not depend on the user, publisher, or advertiser IDs, or on the metadata. More precisely, a positional attribution is parameterized by a class $\mathcal{F} = \{v_m\}_{m \in \N}$ of vectors, where $v_m \in \Delta_{m - 1}$. The attribution function is  defined as $\attr((i_1, \dots, i_m), c) = v_m$. 

It is worth noting that POS is a class of attribution rules, one for each choice of $\mathcal{F}$. The POS class contains FTA ($v_m = (1, 0, \dots, 0)$), LTA ($v_m = (0, \dots, 0, 1)$), UNI ($v_m = (1/m, \dots, 1/m)$) and U-S ($v_m = (0.4, 0.2/(m-2), \dots, 0.2/(m-2), 0.4)$), but it does not contain EXP.

\paragraph{Impression-Priority Attribution (IPA)}
One of the impressions is selected by applying a prioritization function $\pr: \Iids^* \to [0,1]^*$ that depends only on the impressions and not on the conversion, i.e., $\attr((i_1, \dots, i_m), c) = \pr(i_1, \dots, i_m) \in \Delta_{m-1}$.

Similar to POS, IPA is a class of attribution rules; it contains FTA, LTA, UNI, EXP, and U-S.

\section{Differentially Private Conversion Measurement Systems}\label{sec:dp_meas_system}

To describe our framework for DP conversion measurement, we first discuss in Section~\ref{sec:adjacency_and_privacy_scope} the adjacency relations and contribution bounding scopes that we consider, and then describe attribution systems and how to privatize their outputs in Section~\ref{subsec:attr_systems}
\subsection{Adjacency Relations and Contribution Bounding Scopes}\label{sec:adjacency_and_privacy_scope}

As we saw in \Cref{def:dp}, any application of DP should specify a notion of when two datasets are considered adjacent. In the conversion measurement setting, there are several options; the most natural of them are summarized in Table~\ref{fig:adj_table}.

For any relation in the first column of Table~\ref{fig:adj_table}, we can then define two datasets to be \emph{adjacent} if one can be obtained from the other by adding or removing impressions and/or conversions as listed in the second column of the table.

Any application of DP should, at some level, limit the individual contributions; otherwise, the finite amount of noise that is injected would not be sufficient to ensure DP when the individual contributions become too large. In the conversion measurement use case, there are multiple contribution bounding scopes in which the contributions could be limited. These include the same choices listed in Table~\ref{fig:adj_table} for the adjacency relation. We consider henceforth the most natural setting where the contribution bounding scope matches the adjacency relation.  E.g., for the \textsl{user $\times$ publisher} adjacency relation, all the contributions of a given user on a given publisher share the same \budget.

\subsection{Attribution Systems}\label{subsec:attr_systems}

An \emph{attribution system} is an algorithm that takes as input  impressions and conversions sequences and outputs the attributions, represented by a set of weighted pairs of impressions and conversions (defined formally as an \emph{attributed dataset} below). 

\begin{definition}[Attributed Dataset] \label{def:attr-dataset}
An \emph{attributed dataset} $\adat$ is a set of triplets $(i, c, w) \in \Iids \times \Cids \times \R_{\geq 0}$, where for each $(i, c)$ there is a unique $w$.  We may represent an attributed dataset as a function $w_{\adat}: \Iids \times \Cids \to \R_{\geq 0}$, where $w_{\adat}(i, c)$ represents\footnote{In this notation, $w_{\adat}(i, c) = 0$ could correspond to different situations: one is that impression $i$ was considered when attributing conversion $c$ but was not selected by the attribution rule; the other is that impression $i$ and conversion $c$ are incompatible, e.g., impression $i$ takes place after conversion $c$, or they correspond to different users or advertisers.} the total weight attributed to impression $i$ by conversion $c$.
The \emph{$\ell_1$-distance} between two attributed  datasets $\adat, \adat'$ is given by $$\|\adat - \adat'\|_1 := \sum_{i \in \Iids, c \in \Cids} |w_{\adat}(i, c) - w_{\adat'}(i, c)|.$$
\end{definition}
Given an attribution system, we can build a \emph{conversion measurement system} by applying a function $f$ that maps the attributed dataset to a vector in $\R^d$; the vector measures the statistics that an ad tech would like to estimate. For example, if the ad tech wants to know the total attributions for each slice of (campaign $\times$ time-of-day), 
then each of the $d$ dimensions can represent a \valid (campaign ID, time-of-day) pair, and the value that $f$ assigns to that dimension would be the total attribution for that campaign ID and time-of-day.

Of course, as described above, the system is not (differentially) private: the ad tech can allocate a dimension for a particular user and then count exactly, e.g., the number of impressions that user sees.  Since this is not desirable, we employ two methods to ensure privacy. First, we apply \budget enforcement \emph{within the attribution system}, which will be discussed below. Second, we add (appropriately scaled) Laplace noise to each of the $d$ coordinates of the values of $f$; these noisy estimates are then sent to the ad tech.  See \Cref{fig:measurement-system} for an illustration of such a conversion measurement system. Note that we consider the \emph{central} DP setting, where a (trusted) curator runs the attribution rule, computes the function $f$, and adds the noise; the output of this curator is required to be DP.

\tikzstyle{dnode} = [rectangle, rounded corners, minimum width=3cm, minimum height=1cm,text centered, draw=black]
\tikzstyle{arrow} = [thick,->]

\begin{figure}
\centering
\begin{tikzpicture}[node distance=0.75cm]
\node (inputd) [text width=3em] {Input Dataset $\cD$};
\node (rect) [rectangle, rounded corners,minimum width=5cm, minimum height=2cm,text centered, draw=black,xshift=3.6cm] {};
\node (ms) [below of=rect,yshift=-0.5cm] {(DP) Conversion Measurement System};
\node (am) [rectangle, rounded corners,minimum width=2.25cm, minimum height=1cm,text centered, draw=black,right of=inputd,xshift=1.75cm,text width=3em] {Attribution System};
\node (attrd) [right of=am,xshift=1.75cm, text width=3em] {Attributed Dataset $\adat$};
\node (outputd) [right of=attrd,xshift=1.5cm] {$f(\adat) + Z$};
\draw [arrow] (inputd) -- (am);
\draw [arrow] (am) -- (attrd);
\draw [arrow] (attrd) -- (outputd);
\end{tikzpicture}
\caption{Illustration of a (DP) Conversion Measurement System. Each coordinate of the noise $Z$ is drawn from the Laplace distribution with an appropriate scale (see \Cref{lem:main-dp}). We note that the attribution system can include a \budget enforcement component (this is the case in Algorithms~\ref{alg:post-attr} and~\ref{alg:pre-attr}).}
\label{fig:measurement-system}
\end{figure}
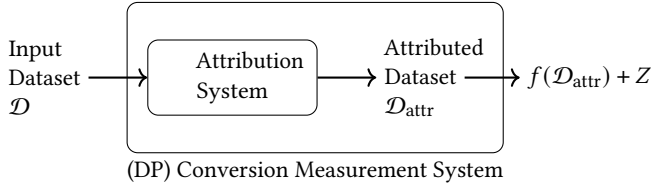

It turns out that there are two important properties needed to ensure DP of the output. The first is that the sensitivity of $f$ is small, i.e., that a small change (in the $\ell_1$-distance) to the attributed dataset does not change the value of $f$ (again, in the $\ell_1$-distance) by much. This is formalized below.

\begin{definition}[Sensitivity of $f$]\label{def:sensitivity}
For a function $f$ that maps an attributed dataset to a vector of real numbers in $\R^d$, we define its \emph{($\ell_1$-)sensitivity} to be
\begin{align*}
\Delta(f) := \max_{\adat, \adat' \atop \|\adat - \adat'\|_1 \leq 1} \|f(\adat) - f(\adat')\|_1.
\end{align*}
\end{definition}

For many natural functions, such as the ``sum by slices'' example above, the sensitivity is bounded (e.g., by $1$ in the example).

The second property we need is with regards to the attribution system itself. Although we have not defined the \budget enforcement yet, it takes in a positive integer parameter $r$, considered as the ``\budget''.\footnote{We note that while in post-attribution enforcement, the \budget $r$ could be set to any positive real number without any modification to our treatment, we choose to keep it integral for consistency with the case of pre-attribution enforcement.}  To ensure DP, we need this parameter $r$ to be an upper bound on the possible change (in the $\ell_1$ sense) in the resulting attributed dataset.

More specifically, an attribution system---which can be specified by a ``configuration'' of adjacency relation, \budget enforcement point, and attribution rule---is ``\valid'' if two adjacent datasets get mapped to attributed datasets that are at most $O(r)$ apart, as stated more precisely below.


\begin{definition}[\Valid Configurations]\label{def:valid_configs}
An adjacency relation along with a \budget enforcement point is said to be \emph{$C_0$-\valid} for a given attribution rule if, for every positive integer $r$, applying a \budget of $r$ at the required enforcement point ensures that any two adjacent datasets always result in two attributed datasets that are at an $\ell_1$-distance of at most $C_0 \cdot r$, where $C_0$ is an absolute constant independent of the numbers of publishers and advertisers.  We call a configuration \emph{valid} if it is $C_0$-valid for some absolute constant $C_0 > 0$.
\end{definition}

Assuming the above two properties, if we appropriately scale the Laplace-distributed noise injected in~\Cref{fig:measurement-system}, then we can guarantee that the system is DP:

\begin{lemma} \label{lem:main-dp}
If the attribution system is instantiated with a $C_0$-\valid configuration and each coordinate of the noise $Z$ is sampled according to the Laplace distribution with scale parameter $C_0 \cdot r \cdot \Delta(f) / \eps$, then the conversion measurement system is $\eps$-DP.
\end{lemma}

We remark that the ``converse'' of Lemma~\ref{lem:main-dp} is also true in the following sense: if we let $f$ be the identity function, i.e., the range of $f$ is associated with $\R^{\cI \times \C}$ and let $f(\adat)_{(i, c)} = w_{\adat}(i, c)$, then the conversion measurement system with $Z$ sampled according to the Laplace distribution with scale $C_0 \cdot r / \eps$ is $\eps$-DP iff the attribution system is instantiated with a $C_0$-\valid configuration. In other words, the \validity of the configuration characterizes the DP property of the conversion measurement system in this sense.

\begin{proof}[Proof of \Cref{lem:main-dp}]
Consider any two adjacent datasets $\cD, \cD'$; let $\adat$ and $\adat'$ be the results of the attribution system on these two datasets, respectively. Then, since the configuration is $C_0$-\valid, we have $\|\adat - \adat'\|_1 \leq C_0 \cdot r$. Therefore, from the definition of $\Delta(f)$, we can conclude that $\|f(\adat) - f(\adat')\|_1 \leq C_0 \cdot r \cdot \Delta(f)$. In other words, the entire measurement system before adding noise is simply a function with $\ell_1$-sensitivity at most $C_0 \cdot r \cdot \Delta(f)$. Thus, the standard DP guarantee of the Laplace mechanism~\cite{dwork06calibrating} implies that the system is $\eps$-DP as desired.
\end{proof}

\begin{remark}
While the treatment above considers the evaluation of a \emph{single} function on the input dataset, it can be readily extended to the setting where one would like to compute  \emph{multiple} (possibly adaptively chosen) functions on the same dataset; this can be done using the standard \emph{composition} properties of DP \cite{dwork2014algorithmic}.
\end{remark}

\begin{remark} \label{remark:typical-queries}
    For typical functions $f$ that are of interest in ad conversion measurement, their sensitivity $\Delta(f)$ can be computed explicitly.
    For example, if $f$ computes the total (attributed) conversion count, then $\Delta(f) = 1$.
    Similarly, if $f$ computes the number of distinct users with attributed conversions, then $\Delta(f) = 1$ too. 
    On the other hand, if $f$ computes the sum of capped conversion values, where each value is capped to a positive real number $V$, then the sensitivity of $f$ is given by $\Delta(f) = V$.

    Queries of interests are also often ``sliced'' by certain attributes. For example, one might be interested in histogram of the total (attributed) conversion count for each publisher or each geographic location (which is e.g. determined by the impression's metadata). In this case, the sensitivity $\Delta(f)$ remains 1.
\end{remark}

\section{\Budget Enforcement}
\label{sec:budget-enforcement}
As stated earlier, a major part of the attribution system is \budget enforcement algorithm. This algorithm takes a positive integer $r$, the ``\budget'', and tries to ``enforce'' this \budget. We consider two types of enforcement in this paper, \emph{pre-attribution} and \emph{post-attribution}, which we will explain next. (It is an interesting future direction to understand if there are other \budget enforcement strategies that may be more privacy safe and/or practical than the ones considered here.)

First, let us note that for the \textsl{conversion} adjacency relation (\Cref{fig:adj_table}), no \budget enforcement is applied. The reason is that each conversion is used only once in the attribution rule. Therefore, the enforcement strategies described below will anyway not affect it.

To describe the enforcement strategies for the other adjacency relations, it is key to define the scope of a contribution bound. 
\begin{definition}[Contribution Bounding Scope]
The \emph{contribution bounding scope} for an adjacency relation is the unit of that  relation.
\end{definition}
For example, for the user adjacency relation, a contribution bounding scope would be each user.

\subsection{Post-Attribution \Budget Enforcement}\label{sec:post_attrib_pb_enforcement}

For post-attribution \budget enforcement, we simply run the attribution algorithm as usual. However, we only add each weighted (impression, conversion) pair to the attributed dataset if it does not exceed the (remaining) \budget of that scope. The pseudo-code is given in \Cref{alg:post-attr}.

\begin{algorithm}[ht]
	\caption{Attribution with Post-Attribution \Budget Enforcement.}
	\label{alg:post-attr}
	\begin{algorithmic}[1]
	    \STATE {\bfseries Parameters:} Attribution rule $\attr$, \budget $r$.
		\STATE {\bfseries Input:} Dataset $\cD = (i_1, \dots, i_n), (c_1, \dots, c_m)$ ordered from oldest to newest.
		\FOR{each contribution bounding scope $s$}
		\STATE $b_s \leftarrow r.$ \hfill \COMMENT{Set remaining \budget of this scope to $r$.}
		\ENDFOR
		\STATE $\adat \leftarrow \emptyset.$ \hfill
		\FOR{$j = 1, \dots, m$}
		\STATE $i'_1, \dots, i'_\ell \leftarrow$ impressions that come before $c_j$ in time and are associated with the same advertiser and the same user as $c_j$.
		\STATE $(w_1, \dots, w_\ell) \leftarrow \attr((i'_1, \dots, i'_\ell), c_j)$. \hfill \COMMENT{Standard attribution alg.} \label{step:attribution-post-attr}
		\FOR{$k=1,\dots,\ell$}
		\STATE $s \leftarrow$ contribution bounding scope corresponding to $(i'_k, c_j, w_k)$.
		\IF{$b_s \geq w_k$}
		\STATE $b_s \leftarrow b_s - w_k$. \hfill \COMMENT{Subtract the spent \budget.}
		\STATE Add $(i'_k, c_j, w_k)$ to $\adat$.
		\ENDIF 
		\ENDFOR
		\ENDFOR
		\RETURN $\adat$
	\end{algorithmic}
\end{algorithm}

\subsection{Pre-Attribution \Budget Enforcement}

Pre-attribution enforcement is much more pessimistic than the post-attribution approach. Specifically, we charge one unit from the \budget of \emph{every scope involved with the input impressions} to the attribution rule. If any scope does not have enough \budget left, we remove all impressions associated with that scope. As we will show below, such a pessimistic approach is---perhaps not too surprisingly---more privacy-safe than the post-attribution approach in certain settings.  The full pseudo-code of the pre-attribution \budget enforcement is given in \Cref{alg:pre-attr}.

\begin{algorithm}[ht]
	\caption{Attribution with Pre-Attribution \Budget Enforcement.}
	\label{alg:pre-attr}
	\begin{algorithmic}[1]
	    \STATE {\bfseries Parameters:} Attribution rule $\attr$, \budget $r$.
		\STATE {\bfseries Input:} Dataset $\cD = (i_1, \dots, i_n), (c_1, \dots, c_m)$ ordered from oldest to newest.
		\FOR{each contribution bounding scope $s$}
		\STATE $b_s \leftarrow r.$ \hfill \COMMENT{Set remaining \budget of this scope to $s$.}
		\ENDFOR
		\STATE $\adat \leftarrow \emptyset.$
		\FOR{$j = 1, \dots, m$}
		\STATE $I \leftarrow$ set of impressions that come before $c_j$ in time and are associated with the same advertiser and the same user as $c_j$ \label{line:all-involved-impressions}
		\STATE $S \leftarrow$ set of contribution bounding scopes corresponding to at least one impression in $I$.
		\FOR{$s \in S$}
		\IF{$b_s \geq 1$}
		\STATE $b_s \leftarrow b_s - 1$. \hfill \COMMENT{Subtract the spent \budget.}
		\ELSE
		\STATE Remove $s$ from $S$. \hfill \COMMENT{Not enough \budget remaining; discard from the final output.}
		\ENDIF 
		\ENDFOR
		\STATE $i'_1, \dots, i'_\ell \leftarrow$ impressions in $I$ corresponding to scopes in $S$. \label{line:select-impression}
		\STATE $(w_1, \dots, w_\ell) \leftarrow \attr((i'_1, \dots, i'_\ell), c_j)$.
		\STATE Add $(i'_1, c_j, w_1), \dots, (i'_\ell, c_j, w_\ell)$ to $\adat$.
		\ENDFOR
		\RETURN $\adat$
	\end{algorithmic}
\end{algorithm}



\section{Classification Results}
\label{sec:our-results}

In this section we present our results on the \validity of different adjacency relations; they are summarized in Table~\ref{tab:conv-level}.  For ease of presentation, we organize our results into two categories: when the \validity holds independent on the attribution rule and otherwise.

We only state in this section the theorems that are proved in Section~\ref{sec:proof}. For the other theorems, we provide forward pointers to their formal statements (along with their proofs) in the Appendix.

\subsection{Attribution Rule-Independent \Validity}

We start by discussing the \validity of different adjacency relations, which turn out to be independent of the attribution rule. The \textsl{conversion} adjacency relation (which, as stated earlier, involves no \budget enforcement) turns out to be valid for every adjacency relation and every attribution rule (Theorem~\ref{thm:conv-valid}).

For pre-attribution enforcement, it also turns out that all adjacency relations and all attribution rules result in \valid configurations (Theorem~\ref{thm:pre-attr}).

We next turn our attention to post-attribution enforcement.  In this case, we show that the validity of the \textsl{user $\times$ advertiser} and \textsl{user} adjacency relations are independent of the attribution rule (Theorems~\ref{thm:post-attr-user} and \ref{thm:post-attr-user-advertiser}). By contrast, the following theorem (which we will prove in Section~\ref{sec:proof}) shows the invalidity of the \textsl{user $\times$ publisher} adjacency relation for any attribution rule.


\begin{restatable}{theorem}{postAttrUserPublisher}
\label{thm:post-attr-user-publisher}
For any attribution rule, the \textsl{user $\times$ publisher} adjacency relation with \budget enforced post-attribution constitutes an invalid configuration.
\end{restatable}

\subsection{Post-Attribution Enforcement and \textsl{Impression} Adjacency}

In this section we consider the \textsl{impression} adjacency under post-attribution enforcement of the \budget.  Our first result proves the \validity of FTA.\footnote{Note that it has a slightly worse absolute constant of $C_0 = 2$ compared to $C_0 = 1$ as in Theorem~\ref{thm:conv-valid}.} Following~\Cref{lem:main-dp}, this means that these configurations require adding twice as much noise compared to those in previous validity results (under the same \budget).

\begin{restatable}{theorem}{postAttrImpressionFta}
\label{thm:post-attr-impression-fta}
For the FTA rule, the \textsl{impression} adjacency relation with \budget enforced post-attribution constitutes a $2$-\valid configuration.
\end{restatable}
%
%
We also prove a similar result for LTA (Theorem~\ref{thm:post-attr-impression-lta}). By contrast, we prove that the UNI, EXP, and U-S attribution rules are all valid in this case (Theorem~\ref{thm:post-attr-impression-uni}, Corollary~\ref{cor:post-attr-impression-exp}, and Theorem~\ref{thm:post-attr-impression-ushaped} respectively).

\subsection{Post-Attribution Enforcement and \textsl{User~$\times$~Publisher~$\times$~Advertiser} Adjacency}

We next consider the \textsl{user $\times$ publisher $\times$ advertiser} adjacency relation. In this case, and under post-attribution enforcement, it turns out that only FTA results in a valid configuration (Theorem~\ref{thm:post-attr-publisher-ad-fta}) whereas the LTA, UNI, EXP, and U-S  attribution rules result in invalid configurations (Theorem~\ref{thm:post-attr-publisher-ad-lta} and Corollary~\ref{cor:post-attr-publisher-ad-uni-exp-ushaped}).

\begin{restatable}{theorem}{postAttrPublisherAdLta}
\label{thm:post-attr-publisher-ad-lta}
For the LTA attribution rule, the \textsl{user $\times$ publisher $\times$ advertiser} adjacency relation with \budget enforced post-attribution constitutes an invalid configuration.
\end{restatable}



\subsection{Intuition and Proof Overview}\label{sec:pf_overview}
Before we proceed to the formal proofs of these classification results, let us briefly (and informally) discuss the high-level ideas behind them. For the invalidity results in the post-attribution case, there are (roughly speaking) two root causes behind them:
\begin{itemize}
\item {\bf Cascading Effect within Attribution Rule.} A single impression can be an input of multiple executions of the attribution algorithm (Line \ref{step:attribution-post-attr} of \Cref{alg:post-attr}). When removing such an impression from the dataset, the attribution rule also changes the weights assigned to other input impressions. Below we construct datasets which make sure that these changes affect different privacy units, implying that it can exceed the contribution bound. This is the idea behind our constructions for the impression adjacency relation (\Cref{thm:post-attr-impression-uni}, \Cref{cor:post-attr-impression-exp}, and \Cref{thm:post-attr-impression-ushaped}). 
\item {\bf Multiple Impressions Affecting Multiple Privacy Units.} In some scenarios (see LTA, FTA discussion below), changing a single impression does not result in a cascading effect. In this case, the high-level idea is to construct multiple impressions---corresponding to the same privacy unit---in such a way that removing each one affects some other different privacy unit. When all of these impressions are removed simultaneously, the effect occurs across different privacy units and therefore bypasses the contribution bounding. This is the gist of our constructions for the user $\times$ publisher adjacency relation (\Cref{thm:post-attr-user-publisher}) and the user $\times$ publisher $\times$ advertiser adjacency relation (\Cref{thm:post-attr-publisher-ad-lta}).
\end{itemize}

We next discuss the validity results. We remark that the attribution rule-independent results (e.g., for the user-level adjacency relation) are relatively straightforward to prove, so we will focus our discussion here on the exceptions: LTA for the impression-level adjacency relation (\Cref{thm:post-attr-impression-lta}), and FTA for the impression-level adjacency (\Cref{thm:post-attr-impression-fta}) and the user $\times$ publisher $\times$ advertiser adjacency relations (\Cref{thm:post-attr-publisher-ad-fta}).
\begin{itemize}
\item {\bf LTA.} When we remove an impression, LTA essentially ``routes'' all the attributions of this impression to the previous one with the same advertiser. Thus, in the impression-level adjacency relation, contribution bounding will upper-bound the change. On the other hand, this reasoning fails for the user $\times$ publisher $\times$ advertiser adjacency relation because it is possible to remove multiple impressions that affect different privacy units (i.e., different publishers); this is indeed the second root cause described above for invalidity.
\item {\bf FTA.} When we remove an impression, if this impression is not the first one of this advertiser, then no change occurs in attribution. Otherwise, FTA  ``routes'' all the attributions of this impression to the second impression of this advertiser. Similar to LTA, this implies validity for the impression-level adjacency relation. However, in contrast to LTA, this argument remains true in the user $\times$ publisher $\times$ advertiser adjacency relation; this is because, even after removing multiple impressions of the same advertiser, all attributions are routed to the same impression---the first one of this advertiser after the removal. Therefore, post-attribution enforcement successfully bounds the change.
\end{itemize}
This concludes our summary of the proof ideas. We will next formalize these by providing the proofs of Theorems~\ref{thm:post-attr-user-publisher},~\ref{thm:post-attr-impression-fta}, and~\ref{thm:post-attr-publisher-ad-lta}. (The remaining proofs are deferred to the Appendix.)

\subsection{Selected Proofs}
\label{sec:proof}
\newcommand{\ti}{\tilde{i}}
\newcommand{\tc}{\tilde{c}}
\newcommand{\tI}{\tilde{I}}
\newcommand{\tC}{\tilde{C}}
\newcommand{\tD}{\tilde{\cD}}

\paragraph{FTA, \textsl{Impression} Adjacency.}
We now prove the validity of the FTA rule for the \textsl{impression} adjacency relation. The proof follows the outline from the previous subsection.

\postAttrImpressionFta*


\begin{proof}[Proof of \Cref{thm:post-attr-impression-fta}]
Consider two adjacent datasets $\cD, \cD'$ such that $\cD'$ results from removing an impression $\ti$. Let $A$ be $\ti$'s advertiser, $U$ be $\ti$'s user and $C_A$ denote the set of conversions from the advertiser and the user. We may assume w.l.o.g. that all conversions in $C_A$ occur \emph{after} the first impression w.r.t. the advertiser $A$ and the user $U$. (As other conversions remain unattributed in both $\cD$ and $\cD'$.)  We consider two cases, based on whether $\ti$ is the first impression w.r.t. its advertiser $A$ and its user $U$ (in $\cD$).
\begin{itemize}
\item Case I: $\ti$ is \emph{not} the first impression w.r.t. the advertiser $A$ and the user $U$. In this case, the two attributed datasets $\adat, \adat'$ are exactly the same, because all conversions in $C_A$ are attributed to the first impression w.r.t. the advertiser $A$ (which is not $\ti$).
\item Case II: $\ti$ is the first impression w.r.t. the advertiser $A$ and the user $U$. In this case, all conversions $c_j \in C_A$ are attributed to $\ti$. These are the only conversions whose attributions change between $\cD$ and $\cD'$.

To analyze this change, consider two subcases, whether $\ti$ is the only impression in $\cD$ from $A$.
\begin{itemize}
\item Case IIa: $\ti$ is the only impression in $\cD$ from $A$ and the user $U$. In this case, all conversions in $C_A$ become unattributed in $\cD'$. Therefore, 
$\|\adat - \adat'\|_1 = \sum_{c_j} w_{\adat}(\ti, c_j) \leq r,$
where the inequality follows from the post-attribution \budget enforcement for 
$\ti$.
\item Case IIb: $\ti$ is not the only impression in $\cD$ from $A$ and the user $U$. Let $\ti'$ be the second impression in $\cD$ from $A$. In this case, every conversion in $C_A$ is either attributed to $\ti'$ or unattributed in $\cD'$. 
Furthermore, no conversions are attributed to $\ti'$ in $\cD$ (because $\ti$ comes before $\ti'$ in the same advertiser $A$ and the same user $U$). 
Therefore, we have
\begin{align*}
\|\adat - \adat'\|_1 = \sum_{c_j} w_{\adat}(\ti, c_j) +  \sum_{c_j} w_{\adat'}(\ti', c_j) \leq 2r,
\end{align*}
where the inequality follows from the post-attribution \budget enforcement for impressions $\ti$ and $\ti'$.
\end{itemize}
\end{itemize}
In all cases, we can conclude that $\|\adat - \adat'\|_1 \leq 2r$. 
\end{proof}

\paragraph{LTA, \textsl{User $\times$ Publisher $\times$ Advertiser} Adjacency.}

Next, we prove the invalidity of the LTA rule under the \textsl{user $\times$ publisher $\times$ advertiser} adjacency relation. 
This is due to the fact that we may arrange the impressions/conversions in such a way that a single publisher gets a large amount of attribution weight (before \budget enforcement) and that, once this publisher is removed, this weight is re-attributed to multiple publishers. The latter ensures that the change grows with the number of publishers. (This is also the main difference between LTA and FTA, since we cannot ensure such a condition for FTA.) Such an example is given together with a formal argument below.

\postAttrPublisherAdLta*

\begin{proof}[Proof of \Cref{thm:post-attr-publisher-ad-lta}]
Let $r = 1$ and let $p > 1$ be any  integer. We construct the dataset $\cD$ as follows:
\begin{itemize}
\item Let there be a single user, a single advertiser, and $p$ publishers $P_1, \dots, P_p$.
\item Let $i_1, \dots, i_{2p - 2}$ be impressions such that impression $i_{2k - 1}$ is associated with publisher $P_k$ and impression $i_{2k}$ is associated with publisher $P_p$ for all $k \in [p - 1]$.
\item Let $c_1, \dots, c_{p - 1}$ be conversions such that $c_k$ appears after $i_{2k}$ and before $i_{2k + 1}$, for all $k \in [p - 1]$.
\end{itemize}
Finally, let $\cD'$ be the dataset resulting from removing publisher $P_p$'s impressions (i.e., $i_2, \dots, i_{2p - 2}$) from $\cD$.

In $\cD$, publishers $P_1, \dots, P_{p - 1}$'s impressions get attributed with zero weight. On the other hand, in $\cD'$, each of these publishers get attribution weight of exactly one. Therefore, $\|\adat - \adat'\|_1 \geq p - 1$, invalidating the attribution system for this configuration.
\end{proof}

\paragraph{Any Attribution Rule, \textsl{User $\times$ Publisher} Adjacency.}
Finally, we prove the invalidity of the \textsl{user $\times$ publisher} adjacency for any attribution rule. We remark that, if we were looking for an invalidity proof of a specific attribution rule, then the construction could have been simplified. For example, the construction in \Cref{thm:post-attr-publisher-ad-lta} above also shows the invalidity of LTA in this setting. However, we would like our proof to generalize to \emph{all} attribution rules. Our construction below accomplishes this by first creating another ``dummy'' dataset $\tD$ (with multiple publishers) to understand how the attribution weights are distributed across different publishers. We then create the datasets $\cD, \cD'$ that differ on the highest weighted publisher to ensure that there is a large--unbounded--change between the two attributed datasets.

\postAttrUserPublisher*


\begin{proof}[Proof of \Cref{thm:post-attr-user-publisher}]
Let $\attr$ be any attribution rule. To create our datasets $\cD, \cD'$, let us start by constructing another dataset $\tD$ as follows.
\begin{itemize}
\item Let there be $\binom{p}{2}$ advertisers $A_{\{1,2\}}, A_{\{1,3\}}, \dots, A_{\{p - 1,p\}}$, and $p$ publishers $P_1, \dots, P_p$.
\item For each advertiser $A_{\{j, k\}}$, let there be impressions $i^{\{j, k\}}_j, i^{\{j, k\}}_k$ and conversion $c^{\{j, k\}}$, coming after both impressions. Furthermore, let $i^{\{j, k\}}_j$ and $i^{\{j, k\}}_k$ be associated with publishers $P_j$ and $P_k$, respectively.
\end{itemize}
Now, suppose we run the attribution system---without any \budget enforcement---on $\tD$. Let $P_\ell$ denote the publisher that gets the largest total attribution weight (with ties broken arbitrarily). Note that the total weight it receives must be at least $\binom{p}{2} / p = 0.5(p - 1)$.

We now construct $\cD$ by keeping only advertisers $A_{\{\ell, j\}}$ for $j \in [p] \setminus \{\ell\}$ in $\tD$ (and discard the rest of advertisers together with all impressions and conversions associated to them). Let the \budget $r$ be 1. Furthermore, let $\cD'$ denote the dataset resulting from removing all impressions corresponding to publisher $P_\ell$ from $\cD$. We will now show that $\|\adat - \adat'\|_1 \geq 0.5(p - 1)$, which implies that the attribution system is an invalid one.

To show this, first observe that, in $\cD'$, each $i^{\{\ell, j\}}_j$ is the only impression that gets fed into the attribution rule for $c^{\{\ell, j\}}$; therefore, it gets attributed with weight one. Furthermore, since we use a \budget $r = 1$ for each (user, publisher) pair and each $i^{\{\ell, j\}}_j$ corresponds to a different publisher, the \budget enforcement leaves these weights unchanged. In summary, we have
\begin{align*}
w_{\adat'}(i^{\{\ell, j\}}_j, c^{\{\ell, j\}}) = 1.
\end{align*}
On the other hand, by our choice of the publisher $P_\ell$, we have
\begin{align*}
\sum_{j \in [p] \setminus \{\ell\}} w_{\adat}(i^{\{\ell, j\}}_j, c^{\{\ell, j\}}) \leq 0.5(p - 1).
\end{align*}
Combining the above two inequalities, we get that
\begin{align*}
&\|\adat - \adat'\|_1 \\
& \geq \sum_{j \in [p] \setminus \{\ell\}} |w_{\adat'}(i^{\{\ell, j\}}_j, c^{\{\ell, j\}}) - w_{\adat}(i^{\{\ell, j\}}_j, c^{\{\ell, j\}})| \\
& \geq 0.5(p - 1). &\qedhere
\end{align*}
\end{proof}

\section{Discussion and Future Directions}\label{conc_future_directions}

In this paper, we presented a formal framework for DP ad conversion measurement setting. We also demonstrated a delicate interplay between attribution and privacy. We defined the notion of operationally valid configurations, and provided a complete classification of the validity of the configurations based on the most popular attribution rules, adjacency relations, contribution bounding scopes, and \budget enforcement points. We hope that our end-to-end differential privacy framework can lead to a solid foundation for practical privacy-preserving ad conversion measurement systems.


While we have focused for simplicity on pure-DP (Definition~\ref{def:dp}), $\ell_1$-sensitivity (Definition~\ref{def:attr-dataset},~\ref{def:sensitivity},~and \ref{def:valid_configs}), and Laplace mechanism (Lemma~\ref{lem:main-dp}), our formalism extends readily to the case of approximate-DP \cite{dwork2006our}, other type of sensitivities, and other DP mechanisms. For example, if the set of measurements is large, we could replace the Laplace mechanism with the partition selection algorithm (see, e.g., \cite{desfontaines2022differentially} and the references therein) if we relax to approximate DP. Similarly, we can extend our framework to $\ell_2$-sensitivity\footnote{Note that changing the sensitivity notion may change the set of valid configurations.} and Gaussian mechanism, which could for instance be used to train DP predicted conversion rate models based on DP stochastic gradient descent \cite{abadi2016deep}. \added{(For prior work on non-private conversion models, see, e.g.,~\cite{agarwal2010estimating, menon2011response, lee2012estimating, rosales2012post}.)}

We describe next some interesting future research directions.

\paragraph{Adjacency Relation $\neq$ Contribution Bounding Scope}
We focused in this work on the most natural setting where the \budget scope is the same as the adjacency relation. 
In principle, this is not necessary: e.g., one might consider a \textsl{user} contribution bounding scope with a \textsl{user $\times$ advertiser} adjacency relation. It might be interesting to give a characterization in such cases, as it will lead to an even more fine-grained understanding of the privacy provided by the conversion measurement system.

\paragraph{Contribution Capping: Beyond Pre- and Post-Attribution?} While we focus on pre-attribution and post-attribution contribution capping, it remains an interesting open question whether there are other (general) capping procedures that can further improve the utility-privacy trade-off.

To illustrate the challenge, note that an intuitive capping strategy is to do it ``at the query evaluation time''. Although such a strategy makes sense for certain query functions $f$ and adjacency relations, it is not completely well-defined for all functions $f$. For example, let $f$ be the number of distinct users with attributed conversions from \Cref{remark:typical-queries} and suppose we are interested in the \textsl{impression} adjacency relation. If two impressions share the same user ID, then it is not clear what their contributions are; on one hand, removing each of them alone does not cause any change to the value of $f$. Meanwhile, removing them both may decrease the value. Such a situation is only exacerbated when we have a more complicated function $f$. We also remark that it is preferable if a single capping procedure is used for all functions $f$ since it allows more flexibility for the measurements that can be made on the platform.

\paragraph{Privacy of the Computation}
Our work has focused on guaranteeing privacy against an adversary that has access to ``what'' is being computed, but not to ``how'' it is being computed.
Concretely, our results capture the setting where a single entity has access to the all the raw impression and conversion data, and seeks to release DP estimates to some requested conversion measurement queries. Studying the privacy of ``how'' this is computed is an important direction for future work. For instance, one could naturally extend this formalism to distributed settings where the trust in a single entity is relaxed by relying on methods such as secure multi-party computing \cite{evans2017pragmatic}, and on-device noise addition as in local DP \cite{evfimievski2003limiting, dwork06calibrating, kasiviswanathan2011can} or shuffle DP \cite{bittau2017prochlo,erlingsson2019amplification,CheuSUZZ19}. Moreover, we studied the case of a static dataset of impressions and conversions; it would be of interest to study the online variant of the problem where privacy needs to be ensured at any time as the impressions and conversions take place.

\paragraph{Enhanced Attribution}
Some attribution systems offer a conversion lookback window option (which limits how far back in time from a conversion are impressions eligible for attribution), and an impression expiry option (which limits how far in the future would the impression be eligible for attribution). It would be interesting to investigate the interplay of these enhancements with privacy and their impact on the validity of the configurations.

In our classification, we considered the simplest and most commonly used attribution rules (listed in Section~\ref{subsubsec:attr_logic}), which operate on a single user's data. It would be interesting to investigate the interplay between DP and more advanced alternatives such as those based on the Shapley value (e.g., \cite{singal2022shapley}) as well as data-driven attribution (DDA), which by contrast is a class of attribution rules that operate on the entire dataset (across users).

\paragraph{Incentives}

While there has been interesting prior work at the intersection of privacy and economics, e.g., \cite{ghosh2011selling, abowd2019economic}, understanding privacy and incentives in the conversion measurement setting would greatly benefit from further investigation.
For instance, our study captures the case where a single ad tech company would like to query the DP conversion measurement system across multiple publishers. In reality, multiple ad techs, which are often competing but could in principle collude, would want to issue DP queries on overlapping impressions and conversions taking place on the same set of publishers and advertisers.

Finally, while a \textsl{user} contribution bounding scope can admit valid configurations, it is vulnerable to ``crowding out attacks'', where, e.g., one publisher can exhaust the \budget of a user (by showing them a large number of impressions). Incorporating the economic incentives of different entities into the analysis of privacy and utility of conversion measurement systems seems worthwhile.

\added{
\paragraph{Privacy-Utility Trade-offs of Various Tasks}
The classification in this work is in terms of sensitivity, which is closely related to additive noise mechanisms as these naturally calibrate the noise scale to the sensitivity. While most proposed DP conversion measurement systems follow this sensitivity and additive noise paradigm, it would be valuable to consider other families of mechanisms, and to quantify the privacy-utility trade-offs of various estimation tasks. 
}

\paragraph{Correlation across Users.}
There are settings where different users' data can be correlated. In ad measurement, this can arise if multiple users watch the same ad (e.g., on a TV). Then, their impressions are correlated, and extra care is needed when applying DP \cite{tschantz2020sok}. We leave the exploration of this interesting setting for future work.

\added{
\paragraph{DP Advertising}
Applying DP in practical advertising systems has been notoriously difficult for the reasons considered in this work, namely: How to define adjacent datasets? Given the correlation between a given user's behavior across (a practically unbounded number of) websites (and/or apps), can DP be applied without adding a disproportionately large amount of noise that would preclude the measurement of simple statistics? We hope that our work provides a stepping stone for tackling these questions in the setting of attribution measurement---the cornerstone of digital advertising---and leads to solid deployments of DP in practical advertising systems.
}

\newpage
\bibliographystyle{ACM-Reference-Format}
\bibliography{refs}

\appendix

\section{Deferred Statements and Proofs}

In this section, we provided all the remaining statements and proofs for our classification results from \Cref{sec:our-results}. We remark that, in all the datasets that we construct below for the invalidity results, all the impressions/conversions belong to a single user, and we will henceforth not specify this explicitly.  Throughout, 
let $\adat$ and $\adat'$ denote the attributed  datasets resulting from $\cD, \cD'$, respectively.

\subsection{\textsl{Conversion} Adjacency Relation}

We start with the simple proof for the validity of the \textsl{conversion} adjacency relation (\Cref{thm:conv-valid}).


\begin{restatable}{theorem}{convValid}
\label{thm:conv-valid}
For any attribution rule, the \textsl{conversion} adjacency relation (without \budget enforcement) constitutes a \valid configuration with $r = C_0 = 1$.
\end{restatable}

\begin{proof}
Consider two adjacent datasets $\cD, \cD'$ such that $\cD'$ results from adding a conversion $c$ to $\cD$.  Notice that $\adat'$ is exactly equal to $\adat$ together with $(i'_1, c, w_1), \dots, (i'_\ell, c, w_\ell)$ where $i'_1, \dots, i'_\ell$ are the impressions that come before $c$ and correspond to the same advertiser as $c$, and $(w_1, \dots, w_\ell) = \attr((i'_1, \dots, i'_\ell), c)$. Therefore, we have $\|\adat - \adat'\|_1 \leq \sum_{j \in [\ell]} w_j = 1,$
where the equality follows from the definition of an attribution rule. Thus, the \textsl{conversion} adjacency relation (without \budget enforcement) constitutes a valid configuration with $r = C_0 = 1$ as desired.
\end{proof}

\subsection{Pre-Attribution Enforcement}

Once again, pre-attribution enforcement results in valid configurations for all adjacency relations and attribution rules.


\begin{restatable}{theorem}{preAttr}
\label{thm:pre-attr}
For any attribution rule, any adjacency relation with \budget enforced pre-attribution constitutes a valid configuration with $C_0 = 1$.
\end{restatable}

\begin{proof}
Consider adjacent datasets $\cD, \cD'$ such that $\cD'$ results from adding impressions $\ti_1, \dots, \ti_a$ and conversions $\tc_1, \dots, \tc_b$ all associated with a single contribution bounding scope $s$. Let $\tI := \{\ti_1, \dots, \ti_a\}$ and $\tC := \{\tc_1, \dots, \tc_b\}$. 

Observe that, in all privacy notions considered, the set of impressions considered for attribution to $c \in \tC$ (on Line~\ref{line:all-involved-impressions} of \Cref{alg:pre-attr}) must come from $\tI$. This means that the attribution rule applied to $\tC$ does not affect any scope apart from $s$. This in turn implies that, in both $\cD, \cD'$, the impressions fed into the attribution rule for each $c_j \notin \tC$ (on Line~\ref{line:select-impression}) are the same apart from those in $\tI$. 

Thus, $\|\adat - \adat'\|_1$ can be at most the number of conversions $c_j$ for which at least one impression from $\tI$ is included in its attribution rule (on Line~\ref{line:select-impression}). However, when this happens, we subtract one from the \budget of scope $s$; therefore, the number of such $c_j$'s can be at most $r$. Hence, we have $\|\adat - \adat'\|_1 \leq r$.
\end{proof}

\subsection{Post-Attribution Enforcement}

We will now move on to the second---and more subtle---privacy enforcement: post-attribution.

\subsubsection{Validity of \textsl{User} and \textsl{User $\times$ Advertiser} Relations}

We start by considering the \textsl{user} adjacency relation.


\begin{restatable}{theorem}{postAttrUser}
\label{thm:post-attr-user}
For any attribution rule, the \textsl{user} adjacency relation with \budget enforced post-attribution constitutes a valid configuration with $C_0 = 1$.
\end{restatable}

\begin{proof}
Consider adjacent datasets $\cD, \cD'$ such that $\cD'$ results from adding impressions $\ti_1, \dots, \ti_a$ and conversions $\tc_1, \dots, \tc_b$ all associated with a single user $U$. Let $\tI := \{\ti_1, \dots, \ti_a\}$ and $\tC := \{\tc_1, \dots, \tc_b\}$.

Notice that outside of the conversions in $\tC$, the attribution system produces exactly the same result for both $\cD, \cD'$. Therefore, we have 
$\|\adat - \adat'\|_1 = \sum_{\ti \in \tI, \tc \in \tC} w_{\adat'}(\ti, \tc).$
Moreover, the post-attribution enforcement exactly ensures that the quantity on the right-hand side is at most $r$. 
\end{proof}


We next consider the \textsl{user $\times$ advertiser} adjacency relation.

\begin{restatable}{theorem}{postAttrUserAdvertiser}
\label{thm:post-attr-user-advertiser}
For any attribution rule, the \textsl{user $\times$ advertiser} adjacency relation with \budget enforced post-attribution constitutes a valid configuration with $C_0 = 1$.
\end{restatable}

\begin{proof}[Proof of \Cref{thm:post-attr-user-advertiser}]
Similar to the proof of \Cref{thm:post-attr-user}, consider two adjacent datasets $\cD, \cD'$ such that $\cD'$ results from adding impressions $\ti_1, \dots, \ti_a$ and conversions $\tc_1, \dots, \tc_b$ all associated with a single user $U$ and a single advertiser $A$. Let $\tI := \{\ti_1, \dots, \ti_a\}$ and $\tC := \{\tc_1, \dots, \tc_b\}$.

Outside of the conversions in $\tC$, the attribution system produces exactly the same result for both $\cD, \cD'$. This implies that $\|\adat - \adat'\|_1 = \sum_{\ti \in \tI, \tc \in \tC} w_{\adat'}(\ti, \tc)$, which is at most $r$ due to post-attribution enforcement.
\end{proof}


\subsubsection{Impression Relation}

We next consider the \textsl{impression} adjacency relation. As stated earlier, the validity in this case does depend on the attribution rule. We note that the validity for FTA was already shown in the main body (\Cref{thm:post-attr-impression-fta}).


\paragraph{Last-Touch Attribution.} 
Again this results in a valid configuration (\Cref{thm:post-attr-impression-lta}).
The idea is similar to before, except that instead of re-attribution to the second impression, the re-attribution happens to the second-to-last impression before the conversion. 


\begin{restatable}{theorem}{postAttrImpressionLta}
\label{thm:post-attr-impression-lta}
For last-touch attribution, the \textsl{impression} adjacency relation with \budget enforced post-attribution constitutes a \valid configuration with $C_0 = 2$.
\end{restatable}

\begin{proof}
Consider adjacency datasets $\cD, \cD'$ such that $\cD'$ results from removing an impression $\ti$. Let $A$ be $\ti$'s advertiser, $U$ be its user, and $C_A$ denote the set of conversions attributed to $\ti$ in $\cD$. We consider two cases, based on whether $\ti$ is the first impression w.r.t. its advertiser $A$ and its user $U$ (in $\cD$).
\begin{itemize}
\item Case I: $\ti$ is the first impression w.r.t. the advertiser $A$ and the user $U$. In this case, all conversions in $C_A$ become unattributed. As a result,
\begin{align*}
\|\adat - \adat'\|_1 = \sum_{c_j} w_{\adat}(\ti, c_j) \leq r,
\end{align*}
where the inequality follows from the post-attribution \budget enforcement for the impression $\ti$.
\item Case II: $\ti$ is \emph{not} the first impression w.r.t. the advertiser $A$ and the user $U$. Let $\ti'$ denote the impression that comes right before $\ti$ w.r.t. the advertiser $A$ and the user $U$. $\ti$ and $\ti'$ are the only two impressions whose attributions change between the two datasets. Furthermore, all attributions to $\ti'$ in $\cD$ remains unchanged in $\cD'$, because the corresponding conversions must come before $\ti$ and therefore before all conversions in $C_A$. In other words, the only changes to $\ti'$ are the additional attributions from conversions in $C_A$. As a result, we have
\begin{align*}
\|\adat - \adat'\|_1 &= \sum_{c_j} w_{\adat}(\ti, c_j) + \sum_{c_j \in C_A} w_{\adat'}(\ti', c_j) \leq 2r,
\end{align*}
where the inequality again follows from the post-attribution \budget enforcement for impressions $\ti$ and $\ti'$.
\end{itemize}
In all cases, we can conclude that $\|\adat - \adat'\|_1 \leq 2r$, which means that this is a valid configuration with $C_0 = 2$ as desired.
\end{proof}

\paragraph{Uniform Multi-Touch and Exponential Time Decay Attributions.}
Next, we move on the prove the invalidity for uniform multi-touch and exponential time decay attribution rules.


\begin{restatable}{theorem}{postAttrImpressionUni}
\label{thm:post-attr-impression-uni}
For uniform multi-touch attribution, the \textsl{impression} adjacency relation with \budget enforced post-attribution constitutes an invalid configuration.
\end{restatable}

Since exponential time decay can be used to implement uniform multi-touch rule when assuming that the time of every impression is the same, \Cref{cor:post-attr-impression-exp} is immediate.

\begin{restatable}{corollary}{postAttrImpressionExp}
\label{cor:post-attr-impression-exp}
For exponential time decay attribution, the \textsl{impression} adjacency relation with \budget enforced post-attribution constitutes an invalid configuration.
\end{restatable}


The idea of the proof of \Cref{thm:post-attr-impression-uni} is to construct an impression that is involved in many conversions' attribution rule. Due to the uniform attribution, such an impression naturally ``draws'' large weights from the overall attribution. Therefore, when removing it, such weights get re-attributed back to the other impressions, resulting in a large change in the attributed dataset.

\begin{proof}
Let $r = 1$ and let $p$ be any positive integer. Consider $\cD$ such that there are $p$ impressions $i_1, \dots, i_p$ and $\binom{p+1}{2}$ conversions $c_{1, 1}, c_{2, 1}, c_{2, 2}, \dots, c_{p, p}$ such that conversion $c_{k, 1}, \dots, c_{k, k}$ appears after $i_k$ (and before $i_{k + 1}$ if it exists); all impressions and conversions are associated with the same advertiser.

Then, let $\cD'$ denote $\cD$ but with $i_1$ removed.

Under the uniform attribution rule, it is simple to see that, for all $j = 2, \dots, p$, we have $w_{\adat}(i_j, c_{j, k}) = \frac{1}{j}$ for all $k = 1, \dots, j$ and $w_{\adat'}(i_j, c_{j, k}) = \frac{1}{j - 1}$ for all $k = 1, \dots, j - 1$. Therefore,
\begin{align*}
\|\adat - \adat'\|_1 
& \geq \sum_{j = 2}^p \sum_{k \in [j]} |w_{\adat}(i_j, c_{j, k}) - w_{\adat'}(i_j, c_{j, k})| \\
& = \sum_{j=2}^p \frac{2}{j} \geq 2 \ln(p/2).
\end{align*}
By taking $p \to \infty$, we can see that this is not bounded above by $C_0 \cdot r$ for any constant $C_0$.
\end{proof}

\paragraph{U-Shaped Attribution.} 
U-shaped attribution rule also results in an invalid configuration (\Cref{thm:post-attr-impression-ushaped}). The construction is an adaptation of the above construction for the uniform attribution. The rough idea is that the ``middle'' part of U-shaped attribution rule is essentially the uniform attribution (scaled by a factor of 0.2). This allows us to use the U-shaped attribution rule to ``implement'' the uniform distribution rule.


\begin{restatable}{theorem}{postAttrImpressionUShaped}
\label{thm:post-attr-impression-ushaped}
For U-shaped attribution, the \textsl{impression} adjacency relation with \budget enforced post-attribution constitutes an invalid configuration.
\end{restatable}

\begin{proof}[Proof of \Cref{thm:post-attr-impression-ushaped}]
Let $r = 1$ and let $p$ be any positive integer greater than three. Consider $\cD$ such that there are $p$ impressions $i_1, \dots, i_p$ and $\binom{p-1}{2} - 1$ conversions, constructed as follows: For every $j = 4, \dots, p$, there are $j - 2$ impressions $c_{j, 1}, \dots, c_{j, j - 2}$ after $i_j$ (and before $i_{j + 1}$ if it exists).
All impressions and conversions are associated with the same advertiser.

Then, let $\cD'$ denote $\cD$ but with $i_1$ removed.

Under the U-shaped attribution rule, 
for all $j = 4, \dots, p - 1$, we have $w_{\adat}(i_j, c_{j+1, k}) = \frac{0.2}{j - 1}$ for all $k = 1, \dots, j - 1$ and $w_{\adat'}(i_j, c_{j, k}) = \frac{0.2}{j - 2}$ for all $k = 1, \dots, j - 2$. Therefore,
\begin{align*}
\|\adat - \adat'\|_1 &\geq \sum_{j = 4}^{p - 1} \sum_{k \in [j - 2]} |w_{\adat}(i_j, c_{j, k}) - w_{\adat'}(i_j, c_{j, k})| \\
&= \sum_{j=4}^{p - 1} \sum_{k \in [j - 2]} \left(\frac{0.2}{j - 2} - \frac{0.2}{j - 1}\right) \\
&\geq \enspace \sum_{j=4}^{p - 1} \frac{0.2}{j-1} 
\enspace \geq \enspace 0.2\ln\left(\frac{p - 1}{4}\right).
\end{align*}
By taking $p \to \infty$, we can see that this is not bounded above by $C_0 \cdot r$ for any constant $C_0$.
\end{proof}

\subsubsection{\textsl{User $\times$ Publisher $\times$ Advertiser} Relation}

We now consider the last adjacency relation: \textsl{User $\times$ Publisher $\times$ Advertiser}.

\paragraph{First-Touch Attribution.}
FTA is the only valid configuration in this case.
The key observation here is that, although removing impressions for a scope may result in re-attribution of multiple conversions, these re-attribution is only to the first impression left for that advertiser. Thereby, an argument analogous to that in the proof of \Cref{thm:post-attr-impression-lta} can show the validity of the configuration.
\begin{restatable}{theorem}{postAttrPublisherAdFta}
\label{thm:post-attr-publisher-ad-fta}
For first-touch attribution, the \textsl{user $\times$ publisher $\times$ advertiser} adjacency relation with \budget enforced post-attribution constitutes a valid configuration with $C_0 = 2$.
\end{restatable}

\begin{proof}[Proof of \Cref{thm:post-attr-publisher-ad-fta}]
Consider two adjacent datasets $\cD, \cD'$ such that $\cD'$ results from removing impressions $\ti_1, \dots, \ti_a$ (ordered from oldest to newest) all associated to a single user $U$, a single publisher $P$ and a single advertiser $A$. Let $C_A$ denote the set of conversions at the advertiser $A$ and the user $U$.

We consider two cases, based on whether $\ti_1$ is the first impression w.r.t. its advertiser $A$ and the user $U$.
\begin{itemize}
\item Case I: $\ti_1$ is \emph{not} the first impression w.r.t. the advertiser $A$ and the user $U$. In this case, the two attributed datasets $\adat, \adat'$ are exactly the same, because all conversions in $C_A$ are attributed to the first impression w.r.t. the advertiser $A$ and the user $U$ (which is not among $\ti_1, \dots, \ti_a$).
\item Case II: $\ti_1$ is the first impression w.r.t. the advertiser $A$ and the user $U$. In this case, all conversions $c_j \in C_A$ are attributed to $\ti_1$. These are the only conversions whose attributions change between $\cD$ and $\cD'$.

To analyze this change, consider two subcases, whether $\ti_1, \dots, \ti_a$ are the only impressions in $\cD$ from the advertiser $A$ and the user $U$.
\begin{itemize}
\item Case IIa: $\ti_1, \dots, \ti_a$ are the only impressions in $\cD$ from the advertiser $A$ and the user $U$. In this case, all conversions in $C_A$ become unattributed in $\cD'$. Thus, $\|\adat - \adat'\|_1 = \sum_{k \in [a]} \sum_{c_j} w_{\adat}(\ti_k, c_j) \leq r,$
where the inequality follows from the post-attribution \budget enforcement for the contribution bounding scope $(U, P, A)$.
\item Case IIb: $\ti_1, \dots, \ti_a$ are not the only impressions in $\cD$ from the advertiser $A$ and the user $U$. Let $\ti'$ be the first impression w.r.t. the advertiser $A$ and the user $U$ that is not among $\ti_1, \dots, \ti_a$. In this case, all conversions in $C_A$ are attributed to $\ti'$ in $\cD'$. 
Furthermore, no conversions are attributed to $\ti'$ (or other impressions in $\ti'$'s contribution bounding scope) in $\cD$ because $\ti_1$ comes first for the same advertiser $A$ and the same user $U$. 
Therefore, we have
\begin{align*}
\|\adat - \adat'\|_1 = \sum_{k \in [a]} \sum_{c_j} w_{\adat}(\ti_k, c_j) +  \sum_{c_j} w_{\adat'}(\ti', c_j) \leq 2r,
\end{align*}
where the inequality again follows from the post-attribution \budget enforcement for the contribution bounding scope $(U, P, A)$ and the contribution bounding scope of $\ti'$.
\end{itemize}
\end{itemize}
In all cases, we can conclude that $\|\adat - \adat'\|_1 \leq 2r$, which means that this is a valid configuration with $C_0 = 2$ as desired.
\end{proof}

\paragraph{Uniform Multi-Touch, Exponential Time Decay and U-Shaped Attributions.}
The invalidity of these attribution rules (\Cref{cor:post-attr-publisher-ad-uni-exp-ushaped}) follows directly from that of the \textsl{impression} case.

\begin{restatable}{corollary}{postAttrPublisherAdUni}
\label{cor:post-attr-publisher-ad-uni-exp-ushaped}
For uniform multi-touch attribution, exponential time decay attribution and U-shaped attribution, the \textsl{user $\times$ publisher $\times$ advertiser} adjacency relation with \budget enforced post-attribution constitutes an invalid configuration.
\end{restatable}

\begin{proof}[Proof of \Cref{cor:post-attr-publisher-ad-uni-exp-ushaped}]
We may use the same constructions as in the proof of \Cref{thm:post-attr-impression-uni}, \Cref{cor:post-attr-impression-exp}, and \Cref{thm:post-attr-impression-ushaped} respectively,  except we assign each $i_1, \dots, i_p$ to $p$ different publishers. This way $\cD, \cD'$ are adjacent under the \textsl{user $\times$ publisher $\times$ advertiser} relation. The remainder of the proof then ensures that $\|\adat - \adat'\|_1$ is not bounded above by $C_0 \cdot r$ for any constant $C_0$.
\end{proof}

\end{document}